\documentclass[11pt,onecolumn,letterpaper]{IEEEtran}
\def\TITReview{1}

\usepackage{amsmath,amssymb,amsfonts,amsthm,mathrsfs}
\usepackage{cite}
\usepackage{booktabs,array}
\usepackage{graphicx}
\usepackage{flafter}
\usepackage{tikz}
\usetikzlibrary{calc,positioning,patterns,decorations.pathreplacing}

\newtheorem{theorem}{Theorem}
\newtheorem{lemma}[theorem]{Lemma}
\newtheorem{proposition}[theorem]{Proposition}
\newtheorem{corollary}[theorem]{Corollary}

\theoremstyle{definition}
\newtheorem{definition}[theorem]{Definition}
\theoremstyle{remark}
\newtheorem{remark}[theorem]{Remark}

\newcommand{\X}{\mathcal{X}}
\newcommand{\F}{\mathcal{F}}
\newcommand{\W}{\mathcal{W}}
\newcommand{\Psh}{\mathcal{P}}
\newcommand{\Ssh}{\mathcal{S}}
\newcommand{\e}{\mathbf{e}}
\DeclareMathOperator{\pref}{pref}
\DeclareMathOperator{\suff}{suff}
\ifdefined\TITReview
  \def\WideMatrixFigure{1}
\newcommand{\matrixfigurewidth}{\textwidth}
\else
  \newcommand{\matrixfigurewidth}{0.98\columnwidth}
\fi

\usepackage[hidelinks,bookmarks=true,bookmarksnumbered=true,bookmarksopen=true,bookmarksopenlevel=1]{hyperref}
\hypersetup{pdftitle={The 3/4 Conjecture for q-Ary Fix-Free Codes With at Most Three Distinct Codeword Lengths},pdfauthor={Weiguo Gao and Zhi Shan}}
\begin{document}

\title{The 3/4 Conjecture for \(q\)-Ary Fix-Free Codes With at Most
Three Distinct Codeword Lengths}

\author{Weiguo~Gao and Zhi~Shan%
\thanks{The authors are with the School of Mathematics,
Fudan University, Shanghai, China.
E-mail: wggao@fudan.edu.cn; 25110180013@m.fudan.edu.cn.}}

\maketitle

\begin{abstract}
We prove the \(3/4\) conjecture for \(q\)-ary fix-free codes with at most
three distinct codeword lengths, for every integer \(q\geq2\).
Every prescribed length distribution with Kraft sum at most \(3/4\)
is realized by a deterministic construction.
We introduce a matrix approach based on an exact identity for the overlap
between forbidden prefix extensions and suffix residuals.
Natural numerical order fixes the shortest layer, and the remaining
selection problem is expressed through row and column counts.
A fixed-cardinality interpolation theorem supplies feasible sets of every
intermediate cardinality between nested endpoints, provided their
differences satisfy one-sided uniqueness and either acyclicity or integral
slack. Reverse-order selection handles the uniform cases directly;
in the remaining cases, layer completion and aligned groups provide
endpoints for interpolation. Together, these methods extend the binary
three-length result to arbitrary finite alphabets and give a deterministic
procedure for constructing the code.
\end{abstract}

\begin{quote}\small\noindent\textbf{Keywords:}
Bifix codes, fix-free codes, Kraft inequality, matrix methods,
variable-length coding.
\end{quote}

\section{Introduction}

A fix-free (or bifix) code has no codeword that is a
proper prefix or suffix of another. The suffix condition permits reverse
parsing and supports applications to bidirectional decoding and error
recovery~\cite{Schuetzenberger,GilbertMoore,Takishima}.

Kraft's theorem characterizes prefix-code length multisets by
\(\sum_i q^{-\lambda_i}\leq1\)~\cite{Kraft}. For fix-free codes, the two forbidden shadows
interact through the actual codewords. Ahlswede, Balkenhol, and
Khachatrian proved the binary \(1/2\) sufficient bound and conjectured
that \(3/4\) is the optimal universal constant~\cite{ABK}.
The \(q\)-ary conjecture asks whether
\begin{equation}
  \sum_{i=1}^k \mu_i q^{-\lambda_i}\leq\frac34,
  \qquad \lambda_1<\cdots<\lambda_k,
  \label{eq:kraft}
\end{equation}
guarantees a fix-free code with \(\mu_i\) words of each length \(\lambda_i\).

Harada and Kobayashi gave a deterministic two-length
construction~\cite{HaradaKobayashi}. Shadow methods, regular systems,
and directed de Bruijn graphs establish further \(q\)-ary
families~\cite{DeppeSchnettler,SchnettlerSurvey}. Binary sufficient
conditions include those of Kukorelly and Zeger~\cite{KukorellyZeger}
and Yekhanin~\cite{Yekhanin04}; approximately \(93.8\%\) of the Kraft-sum-\(3/4\)
vectors considered in~\cite{AghajanKhosravifard} satisfy the latter.
Savari, Yazdi, Abedini, and Khatri treated a restricted binary
three-length case~\cite{SavariEtAl}. Congero and Zeger proved the full
binary three-length theorem by explicit constructions and a random-coding
argument~\cite{CongeroZeger}.

We introduce a deterministic matrix approach that extends the
three-length theorem to every integer \(q\geq2\). Whereas the binary argument
in~\cite{CongeroZeger} averages over prescribed patterns to retain enough
longest words, we fix the shortest layer in natural order and count the
remaining longest words through an exact overlap-matrix identity.
Selecting a middle-layer word enlarges the prefix shadow and reduces
the suffix residual; the matrix formulation tracks both changes.

Theorem~\ref{thm:matrix-interpolation} converts this coupled selection
problem into the construction of nested feasible endpoints, provided
their differences satisfy one-sided uniqueness and either acyclicity
or integral slack. It supplies a deterministic greedy ordering between
the endpoints. In the uniform regimes, reverse order already controls
every cardinality. In the two remaining families, column-layer completion
and phase-aligned groups provide the endpoints to which interpolation
applies. These constructions accommodate the alphabet-dependent block
scales and parity. Section~VI gives the procedure and its implementation cost.

\begin{theorem}[Main result]
\label{thm:main-three-length}
Let \(q\geq2\), let \(1\leq s\leq3\), and let
\(\lambda_1<\cdots<\lambda_s\) and \(\mu_1,\ldots,\mu_s\) be positive
integers.  If
\[
 \sum_{i=1}^s\mu_iq^{-\lambda_i}\leq\frac34,
\]
then there exists a \(q\)-ary fix-free code having exactly \(\mu_i\)
codewords of length \(\lambda_i\) for every \(i\).  Moreover, the proof
gives a deterministic construction determined by the parameters and
the numerical indices of the codewords.
\end{theorem}

Sections II--IV develop extension counts, natural-order selection, and
interpolation. Section~V constructs the required sets in six regimes;
Section~VI completes the proof and gives the algorithm. The appendices
establish the exceptional endpoint estimates.

\section{Extension Counts and the Overlap Matrix}

We count admissible extensions through the overlap of prefix and suffix restrictions. The matrix identity yields the quarter-bound criterion for the construction.

Let \(\F\subseteq\X^*\) be a finite fix-free code over \(\X=\{0,1,\ldots,q-1\}\), with \(q\geq2\), and put
\[
  K(\F)=\sum_{f\in\F}q^{-|f|}.
\]
Write \(r=\max_{f\in\F}|f|\), with \(r=0\) for \(\F=\varnothing\).
For every integer \(n\geq r\), define the length-\(n\)
prefix and suffix shadows by
\begin{align}
  \Psh_n(\F)&=\bigcup_{f\in\F} f\X^{n-|f|},\\
  \Ssh_n(\F)&=\bigcup_{f\in\F} \X^{n-|f|}f.
\end{align}
Fix-freeness makes each shadow union disjoint, with \(q^{n-|f|}\)
words contributed by \(f\). Hence
\begin{equation}
  |\Psh_n(\F)|=|\Ssh_n(\F)|=q^nK(\F).
  \label{eq:shadow-size}
\end{equation}

Let \(\W_n(\F)\) denote the set of length-\(n\) words outside \(\F\) that can be added
to \(\F\) while preserving the fix-free property.

\subsection{The nonoverlap range}

\begin{proposition}[Nonoverlap extension]
\label{prop:nonoverlap}
If \(n\geq2r\), then
\begin{equation}
  |\W_n(\F)|=q^n\bigl(1-K(\F)\bigr)^2.
  \label{eq:nonoverlap-count}
\end{equation}
Consequently, if a length-\(n\) layer of Kraft contribution \(K_n\) is
required and \(K(\F)+K_n\leq3/4\), then
\(|\W_n(\F)|\geq q^nK_n\).
\end{proposition}

\begin{proof}
The length-\(r\) prefix and suffix occupy disjoint coordinates.  Each
has \(q^{r}(1-K(\F))\) admissible values, and the middle block has
\(q^{n-2r}\) values, which proves \eqref{eq:nonoverlap-count}.  Writing
\(x=K(\F)\),
\[
  (1-x)^2-\left(\frac34-x\right)
  =\left(x-\frac12\right)^2\geq0.
\]
Since \(K_n\leq3/4-x\), the second assertion follows.
\end{proof}

Henceforth let the target length \(L\) satisfy
\begin{equation}
  r<L<2r.
  \label{eq:overlap-range}
\end{equation}

\subsection{Matrices in the overlap range}

The matrix construction also applies when \(r\) is any integer upper bound on the codeword lengths. Assume \eqref{eq:overlap-range} and set
\begin{equation}
  h=2r-L,\qquad d=q^h.
  \label{eq:h-d}
\end{equation}
Rows and columns are indexed by the words in \(\X^h\), in base-\(q\)
numerical order.  Let \(\e\in\mathbb R^d\) be the all-ones vector.

\begin{definition}
\label{def:matrices}
Define matrices \(M,A,B\in\mathbb Z_{\geq0}^{d\times d}\), with rows
and columns indexed by \(a,b\in\X^h\), by
\begin{align}
M_{a,b}
 &=\bigl|\{x\in\X^r:\pref_h(x)=a,\ \suff_h(x)=b\}\bigr|,\\
A_{a,b}
 &=\bigl|\{x\in\Psh_r(\F):\pref_h(x)=a,\ \suff_h(x)=b\}\bigr|,\\
B_{a,b}
 &=\bigl|\{x\in\Ssh_r(\F):\pref_h(x)=a,\ \suff_h(x)=b\}\bigr|.
\end{align}
We call \(M,A,B\) the codeword, prefix-shadow, and suffix-shadow
matrices, respectively.  Since \(\Psh_r(\F),\Ssh_r(\F)\subseteq\X^r\),
we have \(0\leq A,B\leq M\) entrywise.  The suffix-residual matrix is
\begin{equation}
  C=M-B.
  \label{eq:C-def}
\end{equation}
\end{definition}

\begin{proposition}[Two forms of the codeword matrix]
\label{prop:two-forms}
Assume \(r<L<2r\) and let \(h=2r-L\).
\begin{enumerate}
\item If \(L\geq3r/2\), equivalently \(2h\leq r\), then
\begin{equation}
  M=q^{r-2h}\e\e^{\mathsf T}.
  \label{eq:dense-M}
\end{equation}
\item If \(L<3r/2\), equivalently \(2h>r\), put
\(s=2h-r=3r-2L\).  Then
\begin{equation}
  M_{a,b}=
  \begin{cases}
  1,&\suff_s(a)=\pref_s(b),\\
  0,&\text{otherwise}.
  \end{cases}
  \label{eq:sparse-M}
\end{equation}
Every row and every column then contains exactly
\(q^{r-h}=q^{L-r}\) ones.
\end{enumerate}
\end{proposition}

\begin{proof}
If \(2h\leq r\), a word with prescribed length-\(h\) prefix \(a\) and
suffix \(b\) has the unique form \(azb\) after the middle word
\(z\in\X^{r-2h}\) is chosen.  This gives \eqref{eq:dense-M}.

If \(2h>r\), the prescribed prefix and suffix overlap in
\(s=2h-r\) symbols.  They determine a length-\(r\) word exactly when
their overlap symbols agree, which gives \eqref{eq:sparse-M}.  For
fixed \(a\), the remaining \(h-s=r-h\) symbols of \(b\) are free; the
column count is the same.
\end{proof}

Fig.~\ref{fig:matrix-examples} illustrates both forms of \(M\) for
length-\(r\) words with length-\(h\) prefix \(a\) and suffix \(b\).

\begin{figure}[!ht]
\centering
\resizebox{\matrixfigurewidth}{!}{\begin{tikzpicture}[
  x=1cm,y=1cm,
  every node/.style={font=\small},
  bit/.style={draw,minimum width=0.62cm,minimum height=0.54cm,inner sep=0pt},
  note/.style={align=center}
]
\node[anchor=west,font=\bfseries] at (0,6.55)
  {(a) \(q=2,\ r=3,\ L=5,\ h=1\)};
\node[bit,fill=black!8]  (d1) at (1.45,5.80) {\(a\)};
\node[bit,fill=black!22] (d2) at (2.07,5.80) {\(z\)};
\node[bit,fill=black!38] (d3) at (2.69,5.80) {\(b\)};
\node[anchor=west] at (3.25,5.80)
  {\(x=azb,\quad a,b,z\in\{0,1\}\)};
\node[anchor=west] at (0.20,4.60) {
  \(\displaystyle
  M=\begin{array}{c|cc}
      & b=0 & b=1\\ \hline
  a=0 & 2 & 2\\
  a=1 & 2 & 2
  \end{array}
  =2\mathbf e\mathbf e^{\mathsf T}\)};

\ifdefined\WideMatrixFigure
\draw[black!35] (8.25,1.8) -- (8.25,6.85);
\else
\draw[black!35] (0,3.30) -- (8.15,3.30);
\fi

\ifdefined\WideMatrixFigure
\begin{scope}[xshift=8.55cm,yshift=3.05cm]
\else
\begin{scope}[yshift=-1.0cm]
\fi
\node[anchor=west,font=\bfseries] at (0,3.50)
  {(b) \(q=2,\ r=3,\ L=4,\ h=2,\ s=1\)};
\node[anchor=east] at (0.66,2.88) {\(a=\)};
\node[bit,fill=black!8]  (s1)  at (1.05,2.88) {\(a_1\)};
\node[bit,fill=black!22] (s2a) at (1.67,2.88) {\(\gamma\)};
\node[anchor=east] at (0.66,2.22) {\(b=\)};
\node[bit,fill=black!22] (s2b) at (1.67,2.22) {\(\gamma\)};
\node[bit,fill=black!38] (s3)  at (2.29,2.22) {\(b_2\)};
\draw[densely dashed,black!55] (s2a.south) -- (s2b.north);
\node[anchor=west] at (3.05,2.88) {\(x=a_1\gamma b_2\)};
\node[anchor=west] at (3.05,2.22) {\(\gamma=a_2=b_1\)};
\node[anchor=west] at (0.10,0.25) {
  \(\displaystyle
  M=\begin{array}{c|cccc}
       & 00&01&10&11\\ \hline
  00&1&1&0&0\\
  01&0&0&1&1\\
  10&1&1&0&0\\
  11&0&0&1&1
  \end{array}\)};
\node[anchor=west] at (5.15,0.25)
  {\(M_{a,b}=1\iff a_2=b_1\)};
\end{scope}
\end{tikzpicture}}
\caption{Exact binary examples of Proposition~\ref{prop:two-forms}.
For \(r=3,L=5\), every pair \((a,b)\) leaves the middle bit \(z\)
free, so \(M_{a,b}=2\).  For \(r=3,L=4\), the prescriptions
\(a=a_1\gamma\) and \(b=\gamma b_2\) share one symbol.  They are
compatible exactly when \(\gamma=a_2=b_1\), and then determine the
unique word \(a_1\gamma b_2\).}
\label{fig:matrix-examples}
\end{figure}

\subsection{The exact extension identity}

\begin{theorem}[Matrix extension identity]
\label{thm:identity}
For \(r<L<2r\),
\begin{equation}
  |\W_L(\F)|
  =q^L\bigl(1-K(\F)\bigr)-\e^{\mathsf T}AC\e.
  \label{eq:extension-identity}
\end{equation}
\end{theorem}

\begin{proof}
There are \(q^L(1-K(\F))\) length-\(L\) words whose length-\(r\)
suffix lies outside \(\Ssh_r(\F)\).  Such a word is inadmissible
exactly when its length-\(r\) prefix belongs to \(\Psh_r(\F)\).
Multiplication of \(A\) and \(C\) matches the length-\(h\) suffix of
the first endpoint word with the length-\(h\) prefix of the second.
Hence \(\e^{\mathsf T}AC\e\) counts precisely these remaining
inadmissible words.
\end{proof}

For \(K_L=\mu q^{-L}\), Theorem~\ref{thm:identity} gives the exact criterion
\begin{equation}
  \e^{\mathsf T}AC\e
  \leq q^L\bigl(1-K(\F)-K_L\bigr).
  \label{eq:exact-criterion}
\end{equation}
Under the global \(3/4\) bound, the following multiplicity-independent
criterion suffices.

\begin{corollary}[Quarter criterion]
\label{cor:quarter}
If \(K(\F)+K_L\leq3/4\) and
\begin{equation}
  \e^{\mathsf T}AC\e\leq\frac{q^L}{4},
  \label{eq:quarter}
\end{equation}
then at least \(q^LK_L\) length-\(L\) words can be added.
\end{corollary}

At total Kraft sum \(3/4\), the quarter criterion is exact.

Define the two profiles
\begin{equation}
  a=A^{\mathsf T}\e,\qquad c=C\e.
  \label{eq:profiles}
\end{equation}
Then \(\e^{\mathsf T}AC\e=a^{\mathsf T}c\), while
\begin{equation}
  \e^{\mathsf T}a=q^rK(\F),\qquad
  \e^{\mathsf T}c=q^r(1-K(\F)).
  \label{eq:profile-sums}
\end{equation}

\begin{proposition}[Profile covariance criterion]
\label{prop:covariance}
If
\begin{equation}
  d\,a^{\mathsf T}c
  \leq(\e^{\mathsf T}a)(\e^{\mathsf T}c),
  \label{eq:covariance}
\end{equation}
then \eqref{eq:quarter} holds.
\end{proposition}

\begin{proof}
Using \(d=q^{2r-L}\) and \eqref{eq:profile-sums},
\[
  a^{\mathsf T}c
  \leq q^LK(\F)(1-K(\F))
  \leq\frac{q^L}{4}.\qedhere
\]
\end{proof}

\section{Natural Order and the Two-Length Case}

The two-length theorem is known~\cite{HaradaKobayashi,DeppeSchnettler},
and has a short probabilistic proof~\cite[Sec.~III]{CongeroZeger}.
The following deterministic proof fixes the shortest layer used later.

Let \(n<L\), let \(\mu,\nu\) be nonnegative integers, and suppose
\begin{equation}
  \mu q^{-n}+\nu q^{-L}\leq\frac34.
  \label{eq:two-kraft}
\end{equation}
For \(x=x_1\cdots x_n\in\X^n\), write
\([x]_q=\sum_{j=1}^n x_jq^{n-j}\), and choose
\begin{equation}
  F_n(\mu)=\{x\in\X^n:[x]_q<\mu\}.
  \label{eq:natural-order}
\end{equation}

\begin{theorem}[Natural-order construction]
\label{thm:natural-order}
For every \(L>n\), the set \(F_n(\mu)\) leaves at least \(\nu\)
admissible length-\(L\) words whenever \eqref{eq:two-kraft} holds.
\end{theorem}

\begin{proof}
Put \(x=\mu q^{-n}\).  If \(L\geq2n\),
Proposition~\ref{prop:nonoverlap} applies directly.

Assume \(n<L<2n\) and put \(h=2n-L\), \(d=q^h\). Use the matrices of Section~II
with \(r=n\). Since
\[
  \Psh_n(F_n(\mu))=\Ssh_n(F_n(\mu))=F_n(\mu),
\]
the shadows coincide:
\begin{equation}
  A=B,\qquad C=M-A.
  \label{eq:two-length-A-equals-B}
\end{equation}
Put \(a=A^{\mathsf T}\e\), \(c=C\e\).
The entry \(a_b\) counts selected words with length-\(h\) suffix \(b\).
The initial numerical interval fills suffix residues cyclically
from \(0\), so \(a\) is nonincreasing.

The length-\(h\) prefix counts \(A\e\) fill complete blocks from left to
right; hence \(A\e\) is nonincreasing. Every row of \(M\) sums to \(q^{L-n}\),
so \(c=(M-A)\e\) is nondecreasing. Rearrangement and
\eqref{eq:profile-sums} give
\begin{align}
  \e^{\mathsf T}AC\e
  =a^{\mathsf T}c
  &\leq\frac{(\e^{\mathsf T}a)(\e^{\mathsf T}c)}{d}\\
  &=\frac{\mu(q^n-\mu)}{q^{2n-L}}
   =q^Lx(1-x)
  \leq\frac{q^L}{4}.
\end{align}
Apply Corollary~\ref{cor:quarter}.
\end{proof}

The proof also gives the terminal saturation estimate
\begin{equation}
  K(F_n(\mu))
  +q^{-L}|\W_L(F_n(\mu))|\geq\frac34.
  \label{eq:natural-order-saturation}
\end{equation}
For \(n<L<2n\), use the extension identity and the preceding bound;
for \(L\geq2n\), Proposition~\ref{prop:nonoverlap} gives
\begin{align*}
  K(F_n(\mu))+q^{-L}|\W_L(F_n(\mu))|
  &=x+(1-x)^2\\
  &=\frac34+\left(x-\frac12\right)^2.
\end{align*}

Fix \(F_n(\mu)\) and \(n\leq r<L<2r\). Form \(A\) and \(C\) for target \(L\) at level \(r\).
Then
\begin{equation}
  \e^{\mathsf T}AC\e\leq\frac{q^L}{4}.
  \label{eq:representation-independent}
\end{equation}
With \(x=\mu q^{-n}\), the extension identity gives
\begin{equation}
  \e^{\mathsf T}AC\e
  =q^L(1-x)-|\W_L(F_n(\mu))|,
  \label{eq:objective-invariant}
\end{equation}
independently of \(r\). The proof of Theorem~\ref{thm:natural-order}
for \(L<2n\), and Proposition~\ref{prop:nonoverlap} for \(L\geq2n\),
bound this expression by \(q^Lx(1-x)\leq q^L/4\).

\section{The Three-Length Selection Problem}

Fixing the shortest layer leaves a middle-layer selection problem: retain enough longest words at the prescribed cardinality. We express this requirement as a matrix bound and prove an interpolation theorem that recovers intermediate cardinalities from suitable endpoints.

Let
\begin{equation}
  \lambda_1<\lambda_2<\lambda_3,\qquad
  \sum_{i=1}^3\mu_iq^{-\lambda_i}\leq\frac34.
  \label{eq:three-kraft}
\end{equation}
Choose
\[
  F_1=F_{\lambda_1}(\mu_1)
\]
in natural order, and let \(\Omega_2\subseteq\X^{\lambda_2}\) be the
set of words compatible with \(F_1\).  By
Theorem~\ref{thm:natural-order},
\begin{equation}
  |\Omega_2|\geq\mu_2.
  \label{eq:omega-size}
\end{equation}

If \(\lambda_3\geq2\lambda_2\), choose any
\(F_2\subseteq\Omega_2\) of cardinality \(\mu_2\), and then apply
Proposition~\ref{prop:nonoverlap} at target length \(n=\lambda_3\).  Thus the
remaining analysis may assume
\begin{equation}
  \lambda_2<\lambda_3<2\lambda_2.
  \label{eq:three-overlap}
\end{equation}
Set
\begin{equation}
  \alpha=\lambda_3-\lambda_2,\qquad
  \beta=2\lambda_2-\lambda_3,\qquad
  \alpha+\beta=\lambda_2.
  \label{eq:alpha-beta}
\end{equation}
The row and column indices of the final codeword matrix have length
\(\beta\).  Moreover,
\begin{equation}
\begin{aligned}
  \alpha\geq\beta
  &\iff \lambda_3\geq\frac32\lambda_2
   &&\text{(dense form)},\\
  \alpha<\beta
  &\iff \lambda_3<\frac32\lambda_2
   &&\text{(sparse form)}.
\end{aligned}
\label{eq:dense-sparse-three}
\end{equation}

For \(F_2\subseteq\Omega_2\), put
\[
  \F(F_2)=F_1\cup F_2.
\]
This union is fix-free because \(F_2\) consists of equal-length words compatible with \(F_1\). Apply Definition~\ref{def:matrices} to \(\F(F_2)\),
at level \(r=\lambda_2\) and target \(L=\lambda_3\), obtaining
\[
  M_3,\quad A_3(F_2),\quad B_3(F_2),\quad
  C_3(F_2)=M_3-B_3(F_2),
\]
in the same order. The notation \(A_3(F_2)\), \(C_3(F_2)\) suppresses the fixed \(F_1\).
Define
\begin{equation}
  \Phi(F_2)
  =\e^{\mathsf T}A_3(F_2)C_3(F_2)\e.
  \label{eq:Phi}
\end{equation}
It suffices to find \(|F_2|=\mu_2\) with
\begin{equation}
  \Phi(F_2)\leq\frac{q^{\lambda_3}}4.
  \label{eq:three-target}
\end{equation}

\subsection{Endpoint states induced by the natural order}

Let
\[
\begin{aligned}
  A_3^\circ&=A_3(\varnothing),&
  C_3^\circ&=C_3(\varnothing),\\
  A_3^*&=A_3(\Omega_2),&
  C_3^*&=C_3(\Omega_2).
\end{aligned}
\]
These are the empty and fully selected middle-layer states.

Apply \eqref{eq:representation-independent} with
\(n=\lambda_1\), \(r=\lambda_2\), \(L=\lambda_3\) to obtain
\begin{equation}
  \e^{\mathsf T}A_3^\circ C_3^\circ\e
  \leq\frac{q^{\lambda_3}}4.
  \label{eq:initial-endpoint}
\end{equation}
The terminal bound is
\begin{equation}
  \e^{\mathsf T}A_3^* C_3^*\e
  \leq\frac{q^{\lambda_3}}4.
  \label{eq:terminal-endpoint}
\end{equation}
Indeed, set \(\F^*=\F(\Omega_2)=F_1\cup\Omega_2\). Since \(\Omega_2=\W_{\lambda_2}(F_1)\), the saturation estimate
\eqref{eq:natural-order-saturation} with \(n=\lambda_1\), \(L=\lambda_2\) gives
\begin{equation}
  K(\F^*)=K(F_1)+q^{-\lambda_2}|\Omega_2|\geq\frac34.
  \label{eq:terminal-kraft}
\end{equation}
Each column of \(M_3\) sums to \(q^\alpha\); from \(0\leq A_3^*\leq M_3\),
\begin{equation}
  (A_3^*)^{\mathsf T}\e\leq q^\alpha\e
  \label{eq:terminal-prefix-profile}
\end{equation}
Also, \eqref{eq:profile-sums} and \eqref{eq:terminal-kraft} give
\begin{equation}
  \e^{\mathsf T}C_3^*\e
  =q^{\lambda_2}(1-K(\F^*))\leq\frac{q^{\lambda_2}}4.
  \label{eq:terminal-residual-mass}
\end{equation}
Using \(C_3^*\e\geq0\),
\[
  \e^{\mathsf T}A_3^*C_3^*\e
  =((A_3^*)^{\mathsf T}\e)^{\mathsf T}(C_3^*\e)
  \leq q^\alpha\e^{\mathsf T}C_3^*\e
  \leq\frac{q^{\lambda_3}}4,
\]
where \(\alpha+\lambda_2=\lambda_3\). This terminal bound holds in all six regimes.

\subsection{Coupled matrix updates}

For \(x\in\Omega_2\), put
\begin{equation}
  i(x)=\pref_\beta(x),\qquad
  j(x)=\suff_\beta(x),\qquad
  E_x=E_{i(x),j(x)}.
\end{equation}
If \(\mathcal Q(F_2)=\sum_{x\in F_2}E_x\), then
\begin{equation}
  A_3(F_2)=A_3^\circ+\mathcal Q(F_2),\qquad
  C_3(F_2)=C_3^\circ-\mathcal Q(F_2).
  \label{eq:coupled-update}
\end{equation}
Each added word increases its prefix entry and decreases its suffix-residual entry by one.

\subsection{Exact increments and interpolation}

The increment formula separates the effect of each new word from interactions within a packet. These interactions determine when feasible endpoints can be joined.

\begin{lemma}[Exact matrix increment]
\label{lem:increment}
Let \(A,C\) be nonnegative \(d\times d\) matrices.  Let \(T\) be a
finite index set and associate with each \(t\in T\) a matrix unit
\(E_{i_t,j_t}\); repetitions are allowed.  Write
\[
  Q_T=\sum_{t\in T}E_{i_t,j_t},\qquad
  A_T=A+Q_T,\qquad C_T=C-Q_T.
\]
Then
\begin{align}
  \e^{\mathsf T}A_TC_T\e
  ={}&\e^{\mathsf T}AC\e
  +\sum_{t\in T}\bigl((C\e)_{j_t}
       -(A^{\mathsf T}\e)_{i_t}\bigr)\nonumber\\
  &-\sum_{s,t\in T}\mathbf{1}_{\{j_s=i_t\}}.
  \label{eq:increment}
\end{align}
\end{lemma}

\begin{proof}
Expand \((A+Q_T)(C-Q_T)\).  The two linear terms give the
profile increments, and
\[
  \e^{\mathsf T}E_{i_s,j_s}E_{i_t,j_t}\e
  =\mathbf{1}_{\{j_s=i_t\}}.
\]
\end{proof}

Since the objective is integral, the relevant threshold is
\begin{equation}
  \mathscr H_3=\left\lfloor\frac{q^{\lambda_3}}4\right\rfloor.
  \label{eq:integer-threshold}
\end{equation}

\begin{lemma}[Interpolation without shared indices]
\label{lem:disjoint-interpolation}
Let \(F^-\subseteq F^+\subseteq\Omega_2\), put
\(D=F^+\setminus F^-\), and suppose
\(\Phi(F^-),\Phi(F^+)\leq\mathscr H_3\).  If
\begin{equation}
  \{i(x):x\in D\}\cap\{j(x):x\in D\}=\varnothing,
  \label{eq:disjoint-indices}
\end{equation}
then, for every \(0\leq m\leq|D|\), there is an \(m\)-element subset
\(S\subseteq D\) such that \(\Phi(F^-\cup S)\leq\mathscr H_3\).
\end{lemma}

\begin{proof}
By \eqref{eq:disjoint-indices}, the last sum in \eqref{eq:increment}
vanishes for every subset. Sort the individual increments increasingly
and take the first \(m\). A nonpositive sum uses the initial endpoint.
Otherwise the \(m\)th increment and all later ones are positive, so
the partial sum is at most the total, controlled by the final endpoint.
\end{proof}

For \(D\subseteq\Omega_2\), let \(\mathcal G(D)\) be the cross-interaction digraph on \(D\), with
an arc \(x\to y\) between distinct \(x,y\in D\) exactly when \(j(x)=i(y)\).
Self-interactions \(i(x)=j(x)\) enter the one-word increment instead.

\begin{theorem}[Fixed-cardinality matrix interpolation]
\label{thm:matrix-interpolation}
Let \(F^-\subseteq F^+\subseteq\Omega_2\), put
\(D=F^+\setminus F^-\), and suppose
\(\Phi(F^-),\Phi(F^+)\leq\mathscr H_3\).  Set
\[
  I=\{i(x):x\in D\},\qquad J=\{j(x):x\in D\}.
\]
For \(S\subseteq D\) and \(x\in D\setminus S\), define the current
increment
\[
  \delta_S(x)=\Phi(F^-\cup S\cup\{x\})-\Phi(F^-\cup S).
\]
Assume one of the following one-sided uniqueness conditions:
\begin{align}
 &\bigl|\{x\in D:j(x)=\ell\}\bigr|\leq1
   &&\text{for every }\ell\in I\cap J,
   \label{eq:column-unique}\tag{U\(_c\)}\\
 &\bigl|\{x\in D:i(x)=\ell\}\bigr|\leq1
   &&\text{for every }\ell\in I\cap J.
   \label{eq:row-unique}\tag{U\(_r\)}
\end{align}
If either
\begin{enumerate}
\item \(\mathcal G(D)\) is acyclic, or
\item the terminal endpoint has one unit of integral slack,
\begin{equation}
  \Phi(F^+)<\mathscr H_3,\qquad\text{equivalently,}\qquad
  \Phi(F^+)\leq \mathscr H_3-1,
  \label{eq:one-unit-slack}
\end{equation}
\end{enumerate}
then the following greedy construction gives a feasible subset of every
cardinality. Set \(S_0=\varnothing\); at each step choose any
\(x_{m+1}\in D\setminus S_m\) minimizing \(\delta_{S_m}(x)\), and put
\(S_{m+1}=S_m\cup\{x_{m+1}\}\). Then
\[
  |S_m|=m,\qquad \Phi(F^-\cup S_m)\leq\mathscr H_3,
  \qquad 0\leq m\leq|D|.
\]
\end{theorem}

\begin{proof}

Writing \(A_S=A_3(F^-\cup S)\) and
\(C_S=C_3(F^-\cup S)\), Lemma~\ref{lem:increment} gives
\begin{equation}
  \delta_S(x)
  =(C_S\e)_{j(x)}-(A_S^{\mathsf T}\e)_{i(x)}
   -\mathbf{1}_{\{i(x)=j(x)\}}.
  \label{eq:current-increment}
\end{equation}

Suppose a greedy step first violates the threshold. Let \(S\) be the preceding set, \(R=D\setminus S\), and write
\[
 n=|R|,\qquad \sigma=\mathscr H_3-\Phi(F^-\cup S)\geq0.
\]
Then \(n\geq1\). If \(x\in R\) is the next choice, minimality and integrality give
\begin{equation}
 \delta_S(y)\geq\delta_S(x)\geq\sigma+1
 \qquad \text{for every }y\in R.
 \label{eq:bad-step-lower-bound}
\end{equation}
The number of arcs in \(\mathcal G(D)[R]\) is
\[
 \chi(R)=\sum_{\substack{x,y\in R\\x\ne y}}
 \mathbf1_{\{j(x)=i(y)\}}.
\]
Applying Lemma~\ref{lem:increment} to \(R\) and separating the self-interactions already counted in \(\delta_S\) yields
\begin{equation}
 \Phi(F^+)-\Phi(F^-\cup S)
 =\sum_{y\in R}\delta_S(y)-\chi(R).
 \label{eq:remainder-increments}
\end{equation}
Condition \eqref{eq:column-unique} bounds every indegree by one;
\eqref{eq:row-unique} bounds every outdegree by one. Thus
\begin{equation}
 \chi(R)\leq n.
 \label{eq:interaction-count}
\end{equation}
If \(\mathcal G(D)\) is acyclic, then \(\chi(R)\leq n-1\): equality in the preceding bound would force a directed cycle. By \eqref{eq:bad-step-lower-bound} and \eqref{eq:remainder-increments},
\[
 \Phi(F^+)-\Phi(F^-\cup S)
 \geq n(\sigma+1)-(n-1)=n\sigma+1>\sigma,
\]
contradicting terminal feasibility. Under \eqref{eq:one-unit-slack}, the same equations instead give
\[
 \sigma-1\geq\Phi(F^+)-\Phi(F^-\cup S)
 \geq n(\sigma+1)-n=n\sigma\geq\sigma,
\]
again a contradiction. Hence every greedy initial segment is feasible.
\end{proof}

\begin{remark}
A first violation without either alternative requires the preceding
value to be \(\mathscr H_3\), every remaining increment to equal one, and
\(\chi(R)=|R|\); the remainder then contains a directed cycle. For nonintegral
\(q^{\lambda_3}/4\), the real bound \(\Phi(F^+)<q^{\lambda_3}/4\) need not give \eqref{eq:one-unit-slack}.
\end{remark}

\section{The Six Length Regimes}
\label{sec:six-regimes}

We now construct the middle-layer set with the quarter bound needed
to complete the longest layer. The geometry is determined by the length
gap \(\alpha=\lambda_3-\lambda_2\) and the matrix-coordinate length
\(\beta=2\lambda_2-\lambda_3\), with \(d=q^\beta\) row and column indices.
We retain the matrix notation of Section~IV. Comparing \(\alpha\) and
\(\beta\) determines whether the support is dense or sparse; the position
of \(\lambda_1\) determines the shortest-layer exclusions.
Table~\ref{tab:six-regimes} records the six possibilities.

\begin{table}[!ht]
\centering
\caption{Six parameter regimes in the overlap range}
\label{tab:six-regimes}
\small
\begin{tabular}{ccl}
\toprule
Regime & \(M_3\) & Length relation\\
\midrule
I   & dense  & \(\lambda_1\leq\beta\leq\alpha\)\\
II  & dense  & \(\beta<\lambda_1\leq\alpha\)\\
III & dense  & \(\beta\leq\alpha<\lambda_1\)\\
IV  & sparse & \(\lambda_1\leq\alpha<\beta\)\\
V   & sparse & \(\alpha<\lambda_1\leq\beta\)\\
VI  & sparse & \(\alpha<\beta<\lambda_1\)\\
\bottomrule
\end{tabular}
\end{table}

Reverse order treats I, II, and IV directly. Column-layer completion
treats III and VI, while phase-aligned groups treat V. These latter
constructions verify the difference-set conditions of
Theorem~\ref{thm:matrix-interpolation} between feasible endpoints.

\subsection{Regimes I, II, and IV: uniform baseline profiles}

Here \(\lambda_1\leq\alpha\) makes both baseline profiles constant.
Reverse selection gives ordered row and column counts, so a single
product estimate controls every cardinality.

\begin{proposition}[Uniform baseline profiles]
\label{prop:uniform-baseline}
In Regimes I, II, and IV,
\begin{equation}
  (A_3^\circ)^{\mathsf T}\e
    =\mu_1q^{\alpha-\lambda_1}\e,
  \qquad
  C_3^\circ\e
    =\bigl(q^\alpha-\mu_1q^{\alpha-\lambda_1}\bigr)\e.
  \label{eq:uniform-baseline-template}
\end{equation}
\end{proposition}

\begin{proof}
At level \(\lambda_2\), the prefix shadow of \(F_1\) is
\(F_1\X^{\lambda_2-\lambda_1}\), and its suffix shadow is
\(\X^{\lambda_2-\lambda_1}F_1\).  Since
\(\lambda_1\leq\alpha=\lambda_2-\beta\), the final
\(\beta\) coordinates of the first set and the initial \(\beta\)
coordinates of the second set lie entirely in their free parts.  Each
length-\(\beta\) coordinate therefore occurs
\(\mu_1q^{\alpha-\lambda_1}\) times.  The row sums of \(M_3\) are
\(q^\alpha\), which gives the second identity.
\end{proof}

Write
\begin{equation}
\begin{aligned}
  \mathcal B_{\mathrm p}
    &=\{x\in\X^{\lambda_2}:\pref_{\lambda_1}(x)\in F_1\},\\
  \mathcal B_{\mathrm s}
    &=\{x\in\X^{\lambda_2}:\suff_{\lambda_1}(x)\in F_1\},\\
  R_2&:=\sum_{x\in\Omega_2}E_x=\mathcal Q(\Omega_2).
\end{aligned}
\label{eq:middle-forbidden-regions}
\end{equation}
Then \(\Omega_2=\X^{\lambda_2}\setminus
(\mathcal B_{\mathrm p}\cup\mathcal B_{\mathrm s})\), and \(R_2(a,b)\) counts admissible positions in cell \((a,b)\).
Fig.~\ref{fig:natural-order-forbidden-regions} displays these exclusions
in the dense supports of Regimes I and II and the sparse support of
Regime IV.

\begin{figure*}[!t]
\centering
\resizebox{\textwidth}{!}{\begin{tikzpicture}[font=\scriptsize]
\tikzset{
  ffcell/.style={draw=black!65,line width=.35pt},
  ffslot/.style={draw=black!45,line width=.22pt},
  ffpref/.style={pattern=north east lines,pattern color=black!58},
  ffsuff/.style={pattern=north west lines,pattern color=black!58},
  ffboth/.style={fill=black!52},
  ffopen/.style={fill=white},
  ffzero/.style={fill=black!7}
}
\def\ffcs{.72}

\begin{scope}[xshift=0cm]
  \node[font=\small] at (2.55,5.30)
    {\(\mathrm{(I)}\quad \lambda_1=1\leq\beta=\alpha=2\)};
  \node at (2.55,4.91)
    {\((\lambda_1,\lambda_2,\lambda_3)=(1,4,6)\)};
  \node at (2.55,4.61) {\(F_1=\{0\}\)};
  \def\ffxo{1.08}\def\ffyt{3.88}
  \node at (2.52,4.29) {\(b=\suff_2(x)\)};
  \node[rotate=90] at (.30,2.64) {\(a=\pref_2(x)\)};
  \foreach \lab [count=\i from 0] in {00,01,10,11}{
    \pgfmathsetmacro{\xx}{\ffxo+(\i+.5)*\ffcs}
    \pgfmathsetmacro{\yy}{\ffyt-(\i+.5)*\ffcs}
    \node at (\xx,4.03) {\lab};
    \node[anchor=east] at (.98,\yy) {\lab};
  }
  \foreach \a in {0,...,3}{
    \foreach \b in {0,...,3}{
      \pgfmathsetmacro{\xx}{\ffxo+\b*\ffcs}
      \pgfmathsetmacro{\yy}{\ffyt-(\a+1)*\ffcs}
      \pgfmathtruncatemacro{\pv}{\a<2 ? 1 : 0}
      \pgfmathtruncatemacro{\sv}{mod(\b,2)==0 ? 1 : 0}
      \ifnum\pv=1
        \ifnum\sv=1
          \path[ffcell,ffboth] (\xx,\yy) rectangle ++(\ffcs,\ffcs);
        \else
          \path[ffcell,ffpref] (\xx,\yy) rectangle ++(\ffcs,\ffcs);
        \fi
      \else
        \ifnum\sv=1
          \path[ffcell,ffsuff] (\xx,\yy) rectangle ++(\ffcs,\ffcs);
        \else
          \path[ffcell,ffopen] (\xx,\yy) rectangle ++(\ffcs,\ffcs);
        \fi
      \fi
    }
  }
  \node at (2.52,.48) {\(\displaystyle
    R_2=\left(\begin{smallmatrix}
    0&0&0&0\\0&0&0&0\\0&1&0&1\\0&1&0&1
    \end{smallmatrix}\right)\)};
\end{scope}

\begin{scope}[xshift=5.35cm]
  \node[font=\small] at (2.55,5.30)
    {\(\mathrm{(II)}\quad \beta=2<\lambda_1=\alpha=3\)};
  \node at (2.55,4.91)
    {\((\lambda_1,\lambda_2,\lambda_3)=(3,5,8)\)};
  \node at (2.55,4.61) {\(F_1=F_3(3)=\{000,001,010\}\)};
  \def\ffxo{1.08}\def\ffyt{3.88}
  \node at (2.52,4.29) {\(b=\suff_2(x)\)};
  \node[rotate=90] at (.30,2.64) {\(a=\pref_2(x)\)};
  \foreach \lab [count=\i from 0] in {00,01,10,11}{
    \pgfmathsetmacro{\xx}{\ffxo+(\i+.5)*\ffcs}
    \pgfmathsetmacro{\yy}{\ffyt-(\i+.5)*\ffcs}
    \node at (\xx,4.03) {\lab};
    \node[anchor=east] at (.98,\yy) {\lab};
  }
  \foreach \a in {0,...,3}{
    \foreach \b in {0,...,3}{
      \pgfmathsetmacro{\yy}{\ffyt-(\a+1)*\ffcs}
      \foreach \t in {0,1}{
        \pgfmathsetmacro{\xx}{\ffxo+\b*\ffcs+.5*\t*\ffcs}
        \pgfmathtruncatemacro{\pv}{2*\a+\t<3 ? 1 : 0}
        \pgfmathtruncatemacro{\sv}{4*\t+\b<3 ? 1 : 0}
        \ifnum\pv=1
          \ifnum\sv=1
            \path[ffslot,ffboth] (\xx,\yy) rectangle ++({.5*\ffcs},\ffcs);
          \else
            \path[ffslot,ffpref] (\xx,\yy) rectangle ++({.5*\ffcs},\ffcs);
          \fi
        \else
          \ifnum\sv=1
            \path[ffslot,ffsuff] (\xx,\yy) rectangle ++({.5*\ffcs},\ffcs);
          \else
            \path[ffslot,ffopen] (\xx,\yy) rectangle ++({.5*\ffcs},\ffcs);
          \fi
        \fi
      }
      \pgfmathsetmacro{\xx}{\ffxo+\b*\ffcs}
      \draw[ffcell] (\xx,\yy) rectangle ++(\ffcs,\ffcs);
    }
  }
  \node at (2.52,.82) {\(x=a\,t\,b,\quad t=0\mid1\)};
  \node at (2.52,.28) {\(\displaystyle
    R_2=\left(\begin{smallmatrix}
    0&0&0&0\\1&1&1&1\\1&1&1&2\\1&1&1&2
    \end{smallmatrix}\right)\)};
\end{scope}

\begin{scope}[xshift=10.70cm]
  \node[font=\small] at (2.55,5.30)
    {\(\mathrm{(IV)}\quad \lambda_1=\alpha=1<\beta=2\)};
  \node at (2.55,4.91)
    {\((\lambda_1,\lambda_2,\lambda_3)=(1,3,4)\)};
  \node at (2.55,4.61) {\(F_1=\{0\}\)};
  \def\ffxo{1.08}\def\ffyt{3.88}
  \node at (2.52,4.29) {\(b=\suff_2(x)\)};
  \node[rotate=90] at (.30,2.64) {\(a=\pref_2(x)\)};
  \foreach \lab [count=\i from 0] in {00,01,10,11}{
    \pgfmathsetmacro{\xx}{\ffxo+(\i+.5)*\ffcs}
    \pgfmathsetmacro{\yy}{\ffyt-(\i+.5)*\ffcs}
    \node at (\xx,4.03) {\lab};
    \node[anchor=east] at (.98,\yy) {\lab};
  }
  \foreach \a in {0,...,3}{
    \foreach \b in {0,...,3}{
      \pgfmathsetmacro{\xx}{\ffxo+\b*\ffcs}
      \pgfmathsetmacro{\yy}{\ffyt-(\a+1)*\ffcs}
      \pgfmathtruncatemacro{\support}{mod(\a,2)==floor(\b/2) ? 1 : 0}
      \ifnum\support=0
        \path[ffcell,ffzero] (\xx,\yy) rectangle ++(\ffcs,\ffcs);
        \node[text=black!35] at ({\xx+.5*\ffcs},{\yy+.5*\ffcs}) {\(\times\)};
      \else
        \pgfmathtruncatemacro{\pv}{\a<2 ? 1 : 0}
        \pgfmathtruncatemacro{\sv}{mod(\b,2)==0 ? 1 : 0}
        \ifnum\pv=1
          \ifnum\sv=1
            \path[ffcell,ffboth] (\xx,\yy) rectangle ++(\ffcs,\ffcs);
          \else
            \path[ffcell,ffpref] (\xx,\yy) rectangle ++(\ffcs,\ffcs);
          \fi
        \else
          \ifnum\sv=1
            \path[ffcell,ffsuff] (\xx,\yy) rectangle ++(\ffcs,\ffcs);
          \else
            \path[ffcell,ffopen] (\xx,\yy) rectangle ++(\ffcs,\ffcs);
          \fi
        \fi
      \fi
    }
  }
  \node at (2.52,.48) {\(\displaystyle
    R_2=\left(\begin{smallmatrix}
    0&0&0&0\\0&0&0&0\\0&1&0&0\\0&0&0&1
    \end{smallmatrix}\right)\)};
\end{scope}

\begin{scope}[yshift=-.55cm]
  \path[ffcell,ffpref] (.20,-.18) rectangle ++(.36,.36);
  \node[anchor=west] at (.66,0) {\(\mathcal B_{\mathrm p}\setminus\mathcal B_{\mathrm s}\)};
  \path[ffcell,ffsuff] (3.25,-.18) rectangle ++(.36,.36);
  \node[anchor=west] at (3.71,0) {\(\mathcal B_{\mathrm s}\setminus\mathcal B_{\mathrm p}\)};
  \path[ffcell,ffboth] (6.30,-.18) rectangle ++(.36,.36);
  \node[anchor=west] at (6.76,0) {\(\mathcal B_{\mathrm p}\cap\mathcal B_{\mathrm s}\)};
  \path[ffcell,ffopen] (9.15,-.18) rectangle ++(.36,.36);
  \node[anchor=west] at (9.61,0) {\(\Omega_2\)};
  \path[ffcell,ffzero] (11.62,-.18) rectangle ++(.36,.36);
  \node at (11.80,0) {\(\times\)};
  \node[anchor=west] at (12.08,0) {\(M_3(a,b)=0\)};
\end{scope}
\end{tikzpicture}}
\caption{Admissible regions induced by the natural-order set \(F_1\) in
representative binary instances of Regimes I, II, and IV.  Hatching
marks the two exclusions, gray their overlap, and white \(\Omega_2\);
crosses in (IV) lie outside the sparse ambient support.  Rows and
columns use numerical order, and \(R_2\) records the cellwise counts.}
\label{fig:natural-order-forbidden-regions}
\end{figure*}
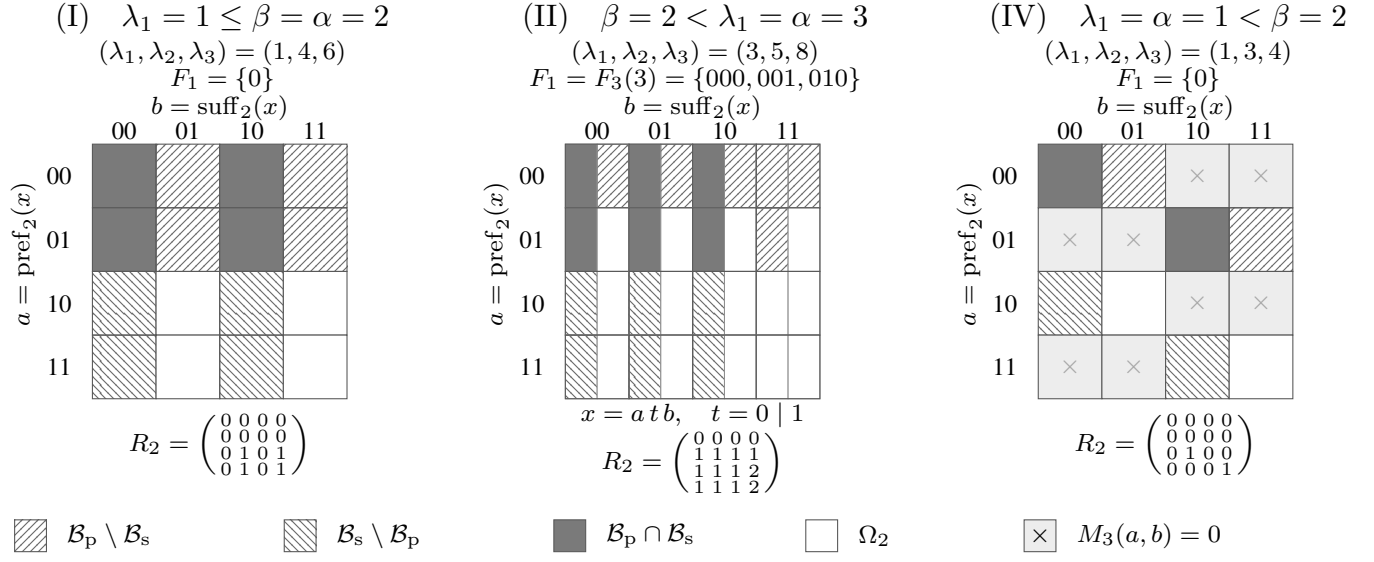

Call a vector \(v\in\mathbb R^d\) \emph{right-loaded} if
\begin{equation}
 \sum_{b=0}^{k-1}v_b\leq\frac{k}{d}\sum_{b=0}^{d-1}v_b,
 \qquad 1\leq k<d.
 \label{eq:right-loaded}
\end{equation}
\paragraph*{Order property of reverse truncation}
Let \(n<r\), \(1\leq h\leq r-n\), \(d=q^h\), and \(0\leq\mu\leq q^n\). Put
\begin{equation}
 \Omega=\{x\in\X^r:\pref_n(x)\notin F_n(\mu),\ \suff_n(x)\notin F_n(\mu)\}.
 \label{eq:reverse-omega}
\end{equation}
For \(0\leq m\leq|\Omega|\), let \(S_m\) consist of the \(m\) numerically largest words of \(\Omega\), with column and row counts
\begin{align}
 p_m(b)&=|\{x\in S_m:\suff_h(x)=b\}|,\nonumber\\
 s_m(a)&=|\{x\in S_m:\pref_h(x)=a\}|.
 \label{eq:reverse-profiles}
\end{align}
Then \(p_m\) is right-loaded and \(s_m\) is nondecreasing.
Indeed, in numerical coordinates the admissibility conditions are
\begin{equation}
 \pref_n(x)\geq\mu,\qquad\suff_n(x)\geq\mu.
 \label{eq:reverse-thresholds}
\end{equation}
Since \(r-h\geq n\), the sets \(\Omega_a=\{y\in\X^{r-h}:ay\in\Omega\}\) are nested increasingly in \(a\): the suffix condition is independent of \(a\), and the prefix condition becomes weaker as \(a\) increases. Reverse order fills these sets from the largest index, cutting at most one, so \(s_m\) is nondecreasing.

For the column counts, put \(Q=q^n\). The prefix cutoff is a multiple of \(q^{r-n}\), hence of \(d\), and the suffix condition retains \([\mu,Q)\) in each \(Q\)-period. Since \(Q\) and \(d\) are \(q\)-powers, one divides the other. If \(d\mid Q\), each nonempty \(d\)-cycle retains a full cycle or a terminal interval. If \(Q\mid d\), each full \(d\)-cycle retains
\[
 T=\bigcup_{\ell=0}^{d/Q-1}[\ell Q+\mu,(\ell+1)Q).
\]
For \(k=tQ+u\), \(0\leq u<Q\),
\[
 |T\cap[0,k)|=t(Q-\mu)+\max\{0,u-\mu\}
 \leq\frac{k}{d}|T|.
\]
Thus the indicator of every retained cycle is right-loaded. This property is preserved when only its numerically largest elements are kept: if \(U\) is such a subset of a set \(T\subseteq[0,d)\) with right-loaded indicator and \(U\cap[0,k)\ne\varnothing\), then
\[
 |U\cap[0,k)|=|U|-|T\cap[k,d)|
 \leq |U|-\frac{d-k}{d}|T|\leq\frac{k}{d}|U|.
\]
The empty-intersection case is immediate. Since the reverse scan consists of full retained cycles and at most one such terminal subset, adding their counts proves the claim for \(p_m\).

\begin{theorem}[Reverse natural order in the uniform regimes]
\label{thm:reverse-uniform-regimes}
In Regimes I, II, and IV, order \(\Omega_2\) by decreasing base-\(q\)
numerical value and let \(F_2(m)\) be its first \(m\) words.  Then, for
every \(0\leq m\leq|\Omega_2|\),
\begin{equation}
  \Phi(F_2(m))\leq\frac{q^{\lambda_3}}4.
  \label{eq:reverse-quarter}
\end{equation}
In particular, \(F_2(\mu_2)\) is a valid middle-layer choice.
\end{theorem}

\begin{proof}
Put \(d=q^\beta\),
\[
  \rho=\mu_1q^{\alpha-\lambda_1},\qquad
  \sigma=q^\alpha-\rho,
\]
and write \(\mathcal Q_m=\mathcal Q(F_2(m))\).  By
Proposition~\ref{prop:uniform-baseline} and
\eqref{eq:coupled-update},
\begin{equation}
\begin{aligned}
  a_m&=A_3(F_2(m))^{\mathsf T}\e=\rho\e+p_m,\\
  c_m&=C_3(F_2(m))\e=\sigma\e-s_m,
\end{aligned}
  \label{eq:reverse-profile-decomposition}
\end{equation}
where \(p_m=\mathcal Q_m^{\mathsf T}\e\),
\(s_m=\mathcal Q_m\e\), and
\(\e^{\mathsf T}p_m=\e^{\mathsf T}s_m=m\).  The common constant terms cancel, giving
\begin{equation}
 a_m^{\mathsf T}c_m-
 \frac{(\e^{\mathsf T}a_m)(\e^{\mathsf T}c_m)}d
 =\frac{m^2}{d}-p_m^{\mathsf T}s_m.
 \label{eq:uniform-objective-expansion}
\end{equation}
By Proposition~\ref{prop:covariance}, it suffices to show
\begin{equation}
 p_m^{\mathsf T}s_m\geq\frac{m^2}{d}.
 \label{eq:uniform-product-expansion}
\end{equation}
Apply the order property with \((n,r,h)=(\lambda_1,\lambda_2,\beta)\), using \(\lambda_1\leq\alpha=\lambda_2-\beta\). Then \(p_m\) is right-loaded and \(s_m\) is nondecreasing. With \(Z_k=\sum_{b=0}^{k-1}(p_m(b)-m/d)\leq0\), Abel summation gives
\begin{equation}
 p_m^{\mathsf T}s_m-\frac{m^2}{d}
 =\sum_{k=1}^{d-1}Z_k\bigl(s_m(k-1)-s_m(k)\bigr)\geq0.
 \label{eq:reverse-abel}
\end{equation}
This proves \eqref{eq:reverse-quarter} for every \(m\); \eqref{eq:omega-size} permits \(m=\mu_2\).
\end{proof}

\subsection{Regimes III and VI: column-layer completion}

The initial column profile is a constant vector plus a head indicator. Selecting a suitable layer tail makes it constant; full admissible layers then preserve this form. Theorem~\ref{thm:matrix-interpolation} fills the cardinalities within each layer. We first establish these facts, then treat the row-block transition and terminal residual.

\subsubsection{Baseline profiles and admissible column layers}

Both regimes satisfy
\[
  \lambda_1>\max\{\alpha,\beta\}.
\]
To describe the free extension and its overlap with the shortest layer, put
\begin{equation}
  t=\lambda_2-\lambda_1,\qquad
  \gamma=\lambda_1-\alpha,\qquad
  \theta=\lambda_1-\beta.
  \label{eq:staircase-scales}
\end{equation}
Then \(t,\gamma,\theta>0\), \(\beta=\gamma+t\), and \(\alpha=\theta+t\).
Here \(t\) is the free extension length; \(\gamma\) and \(\theta\) specify the two overlaps.

\begin{proposition}[Staircase baseline profiles]
\label{prop:staircase-baseline}
Write
\begin{equation}
  \mu_1=\ell q^\gamma+\nu
       =\xi q^\alpha+\eta,
  \qquad
  0\leq\nu<q^\gamma,\quad 0\leq\eta<q^\alpha,
  \label{eq:staircase-divisions}
\end{equation}
Put \(m=\nu q^t\), the width of the elevated column head, and define
\(u_m=\bigl(\mathbf1_{\{j<m\}}\bigr)_{j=0}^{d-1}\).
Then
\begin{equation}
  (A_3^\circ)^{\mathsf T}\e=\ell\e+u_m.
  \label{eq:staircase-prefix-profile}
\end{equation}
Moreover, every row index has a unique representation \(zp\), where
\(z\in\X^t\) and \(p\in\X^\gamma\), and
\begin{equation}
  (C_3^\circ\e)_{zp}=
  \begin{cases}
    0, & [p]_q<\xi,\\
    q^\alpha-\eta, & [p]_q=\xi,\\
    q^\alpha, & [p]_q>\xi.
  \end{cases}
  \label{eq:staircase-residual-profile}
\end{equation}
\end{proposition}

\begin{proof}
A prefix-shadow word is uniquely \(fu\), with \(f\in F_1\), \(u\in\X^t\), and
length-\(\beta\) suffix \(\suff_\gamma(f)u\). Thus column \(pu\) counts integers in \(F_1\)
congruent to \([p]_q\) modulo \(q^\gamma\). By \eqref{eq:staircase-divisions},
this count is \(\ell+1\) for \([p]_q<\nu\) and \(\ell\) otherwise. The first case
comprises the first \(\nu q^t=m\) columns, proving
\eqref{eq:staircase-prefix-profile}.

A suffix-shadow word is uniquely \(uf\), with length-\(\beta\) prefix \(u\pref_\gamma(f)\).
Its shadow row sum at \(up\) is \(q^\alpha\), \(\eta\), or \(0\) according as \([p]_q<\xi\),
\([p]_q=\xi\), or \([p]_q>\xi\). Every row of \(M_3\) sums to \(q^\alpha\); subtracting these
counts from \(M_3\e\) proves \eqref{eq:staircase-residual-profile}.
\end{proof}

Write each length-\(\lambda_2\) word uniquely as \(x=yz\),
\(y\in\X^\alpha\), \(z\in\X^\beta\). Fixing \(y\) gives a column
layer that traverses every column once:
\(\mathcal L_y=\{yz:z\in\X^\beta\}\).
We use the clipping notation
\([a]_0^N=\min\{N,\max\{0,a\}\}\).

\begin{lemma}[Admissible column-layer decomposition]
\label{lem:admissible-column-layers}
For \(y\in\X^\alpha\), the prefix and suffix restrictions give the cutoffs
\begin{equation}
\begin{aligned}
  P_y&=q^t[\mu_1-[y]_q q^\gamma]_0^{q^\gamma},\\
  S_y&=[\mu_1-[\suff_\theta(y)]_q q^\beta]_0^{q^\beta},\\
  T_y&=\max\{P_y,S_y\}.
\end{aligned}
  \label{eq:layer-thresholds}
\end{equation}
Then
\begin{equation}
  \mathcal L_y\cap\Omega_2
  =\{yz:[z]_q\geq T_y\}.
  \label{eq:layer-tail}
\end{equation}
Thus the admissible part of each layer is empty, a full layer, or a
single terminal interval of columns.

If
\(\mathcal Q_y=\sum_{z\in\X^\beta}E_{\pref_\beta(yz),\,z}\),
then
\begin{equation}
  \mathcal Q_y^{\mathsf T}\e=\e.
  \label{eq:unit-column-layer}
\end{equation}
Let \(e_a\) denote the standard basis vector indexed by \(a\).  In
Regime III, with \(r_y=\pref_\beta(y)\),
\(\mathcal Q_y=e_{r_y}\e^{\mathsf T}\).
In Regime VI, with \(s=\beta-\alpha\),
\begin{equation}
  \mathcal Q_y=\sum_{z\in\X^\beta}
       E_{\,y\pref_s(z),\,z},
  \label{eq:sparse-column-layer}
\end{equation}
and every unit in \eqref{eq:sparse-column-layer} lies in the support of
\eqref{eq:sparse-M}.
\end{lemma}

\begin{proof}
Since \(\lambda_1=\alpha+\gamma=\theta+\beta\),
\begin{equation}
  \pref_{\lambda_1}(yz)=y\pref_\gamma(z),\qquad
  \suff_{\lambda_1}(yz)=\suff_\theta(y)z.
  \label{eq:layer-prefix-suffix}
\end{equation}
The first word lies outside the natural-order initial interval \(F_1\)
if and only if
\[
  [y]_q q^\gamma+[\pref_\gamma(z)]_q\geq\mu_1.
\]
Because \(\beta=\gamma+t\), this condition is equivalent to
\([z]_q\geq P_y\).  The second word lies outside \(F_1\) if and only if
\([z]_q\geq S_y\).  Their intersection proves
\eqref{eq:layer-tail}.

The column index of \(yz\) is \(z\), so every column occurs exactly once
as \(z\) runs through \(\X^\beta\); this proves
\eqref{eq:unit-column-layer}.  If \(\beta\leq\alpha\), the row index is
the constant \(\pref_\beta(y)\).  If \(\alpha<\beta\), it is
\(y\pref_{\beta-\alpha}(z)\), and
\(\suff_{\beta-\alpha}(y\pref_{\beta-\alpha}(z))=\pref_{\beta-\alpha}(z)\),
which is exactly the compatibility condition for the sparse support.
\end{proof}

\subsubsection{Completion and interpolation within a layer}
A layer is \emph{fully admissible} if \(T_y=0\); \eqref{eq:layer-thresholds} gives
\begin{equation}
  T_y=0
  \quad\Longleftrightarrow\quad
  \begin{cases}
  [y]_q\geq\lceil\mu_1/q^\gamma\rceil,\\
  [\suff_\theta(y)]_q\geq\lceil\mu_1/q^\beta\rceil.
  \end{cases}
  \label{eq:full-layer-criterion}
\end{equation}
Call \(\mathcal L_y\)
\emph{completion-capable} if
\begin{equation}
  T_y\leq m.
  \label{eq:completion-capable-layer}
\end{equation}
Then \eqref{eq:layer-tail} makes the entire tail \(\{yz:m\leq[z]_q<q^\beta\}\) admissible.
Selecting this tail adds \(\e-u_m\) to the column profile and makes it constant. The admissible head \(\{yz:T_y\leq[z]_q<m\}\) is reserved for the terminal residual.

\begin{lemma}[Acyclic interpolation within a column layer]
\label{lem:column-layer-interpolation}
Fix \(y\in\X^\alpha\) and let
\(D\subseteq\{yz:z\in\X^\beta\}\cap\Omega_2\).
Then \(D\) satisfies \eqref{eq:column-unique} and its cross-interaction
digraph is acyclic. If \(F^-\subseteq F^-\cup D\subseteq\Omega_2\)
have feasible endpoints, Theorem~\ref{thm:matrix-interpolation}
supplies every intermediate cardinality.
\end{lemma}
\begin{proof}
Column indices are unique. The row map
\(\mathcal T_y(z)=\pref_\beta(yz)\) is nondecreasing in numerical order
and thus has no nontrivial periodic orbit. An arc \(z\to w\) means
\(z=\mathcal T_y(w)\), so a directed cycle on distinct vertices would
be such an orbit. Theorem~\ref{thm:matrix-interpolation} applies.
\end{proof}

If
\(A_3(F_2)^{\mathsf T}\e=s\e\)
and \(K=K(F_1\cup F_2)\), the profile sums and \(q^\beta=q^{2\lambda_2-\lambda_3}\) give
\begin{equation}
 \Phi(F_2)=s\e^{\mathsf T}C_3(F_2)\e
 =q^{\lambda_3}K(1-K)\leq\frac{q^{\lambda_3}}4.
 \label{eq:constant-profile-anchor}
\end{equation}
Thus constant prefix profiles are feasible anchors.

\begin{theorem}[Continuous cardinalities from a completion layer]
\label{thm:complete-layer-core}
Let
\[
 \begin{gathered}
  \mathcal Y_0=\{y\in\X^\alpha:T_y=0\},\\
  \mathcal Y_{\rm c}=\{y\in\X^\alpha:T_y\leq m\}.
 \end{gathered}
\]
Choose \(y_0\in\mathcal Y_{\rm c}\), and put
\(L_0=q^\beta-m+|\mathcal Y_0\setminus\{y_0\}|q^\beta\).
Then there is a nested sequence
\[
  \varnothing=F_2^{(0)}\subset F_2^{(1)}
  \subset\cdots\subset F_2^{(L_0)}\subseteq\Omega_2
\]
such that, for every \(0\leq k\leq L_0\),
\begin{equation}
  |F_2^{(k)}|=k,\qquad
  \Phi(F_2^{(k)})\leq\frac{q^{\lambda_3}}4.
  \label{eq:complete-layer-core-bound}
\end{equation}
If \(\mu_1\leq q^{\lambda_1}/2\), then
\(\mathcal Y_0\subseteq\mathcal Y_{\rm c}\) is nonempty.
\end{theorem}

\begin{proof}
Select \(D_0=\{y_0z:m\leq[z]_q<q^\beta\}\). By
\eqref{eq:completion-capable-layer}, this tail is admissible;
\eqref{eq:unit-column-layer} and \eqref{eq:staircase-prefix-profile} give
\[
 A_3(D_0)^{\mathsf T}\e
 =\ell\e+u_m+(\e-u_m)=(\ell+1)\e.
\]
The initial and completed endpoints are feasible by
\eqref{eq:initial-endpoint} and \eqref{eq:constant-profile-anchor}.
Add the layers in \(\mathcal Y_0\setminus\{y_0\}\) one at a time.
Each adds \(\e\), preserving the constant column profile. Lemma~\ref{lem:column-layer-interpolation}
fills every step, yielding a nested chain of size
\(|D_0|+|\mathcal Y_0\setminus\{y_0\}|d=L_0\).

If \(\mu_1\leq q^{\lambda_1}/2\), take \(y=(q-1)^\alpha\).
Since \(\alpha,\theta\geq1\),
\[
\begin{aligned}
 {}[y]_q q^\gamma&=q^{\lambda_1}-q^\gamma
       \geq q^{\lambda_1}/2\geq\mu_1,\\
 [\suff_\theta(y)]_q q^\beta&=q^{\lambda_1}-q^\beta
       \geq q^{\lambda_1}/2\geq\mu_1.
\end{aligned}
\]
Thus \(y\in\mathcal Y_0\) by \eqref{eq:full-layer-criterion}.
\end{proof}

\subsubsection{Reverse progression and row-block transitions}

Complete \(g\)-groups organize reverse progression through partial layers; transitions between row blocks are handled separately.
Use the counting scales
\[
 H=q^\alpha,\qquad g=q^\gamma,\qquad B=q^t,\qquad A=q^\theta.
\]
Then \(H=AB\), \(d=Bg\), and \(Q=q^{\lambda_1}=Ad=Hg\).
There are \(B\) blocks of \(g\) rows and \(A\) standard column layers per block. Retaining \(\mu_1=\ell g+\nu\), write
\[
 \mu_1=kd+R,\quad R=cg+\nu,\quad m=B\nu,\qquad 0\le R<d.
\]
Here \(k,R\) locate the suffix cutoff. Identify words with their base-\(q\) values; then
\(i(w)=\lfloor w/H\rfloor\), \(j(w)=w\bmod d\), and
\begin{equation}
 i(w)\in[bg,(b+1)g)\ \Longleftrightarrow\ w\in[bQ,(b+1)Q).
 \label{eq:aligned-row-block}
\end{equation}
The layers in block $b$ have indices $bA,\ldots,(b+1)A-1$.
Write $\mathcal L_y[a,b)=\{yd+j:a\le j<b\}$.
Formula~\eqref{eq:layer-thresholds} becomes
\begin{equation}
 T_y=\max\{B[\mu_1-yg]_0^g,
                 [\mu_1-(y\bmod A)d]_0^d\}.
 \label{eq:local-layer-threshold}
\end{equation}
Write \(a,c_{\rm current}\) for the current column and residual row profiles, \(c^\circ=C_3^\circ\e\), and \(|v|=\e^{\mathsf T}v\) for vector mass. For a row-count vector \(V\), let \(V(J)=\sum_{i\in J}V_i\). A complete \(g\)-group in a layer has columns \([jg,(j+1)g)\).

Figure~\ref{fig:staircase-geometry} shows the geometric decomposition into a completion tail, full layers, complete groups, and residual fragments.

\begingroup
\newcommand{\ffcoltitle}{\(j=\suff_\beta(x)\)}
\newcommand{\fflayerrow}{\(y=\pref_\alpha(x)\)}
\newcommand{\ffsparserow}{\(i=\pref_\beta(x)\)}
\newcommand{\ffforbidden}{$F_1$ exclusions}
\newcommand{\ffcompletionlegend}{$\mathcal S_{\rm c}$: completion/full layers}
\newcommand{\ffgroupslegend}{$\mathcal S_g$: complete groups}
\newcommand{\ffresiduallegend}{$\mathcal R$: residual}
\ifdefined\TITReview
\begin{figure}[!ht]
\else
\begin{figure*}[!t]
\fi
\centering
\resizebox{0.88\textwidth}{!}{\begin{tikzpicture}[font=\scriptsize]
\tikzset{
  ffcell/.style={draw=black!70,line width=.34pt},
  ffforbid/.style={fill=black!62},
  ffcompletion/.style={pattern=north east lines,pattern color=black!62},
  ffgroups/.style={pattern=crosshatch,pattern color=black!58},
  ffresidual/.style={pattern=dots,pattern color=black!68},
  ffzero/.style={fill=black!7}
}

\begin{scope}[xshift=0cm]
  \node[font=\small] at (3.10,7.05)
    {\(\mathrm{(III)}\quad(\lambda_1,\lambda_2,\lambda_3)=(3,4,6)\)};
  \node at (3.10,6.65)
    {\(\mu_1=1,\ d=4,\ g=2,\ m=2,\ y_0=00\)};
  \def\xo{1.35}\def\yt{5.18}\def\cs{.88}
  \node at (3.11,5.78) {\ffcoltitle};
  \node[rotate=90] at (.36,3.42) {\fflayerrow};
  \foreach \lab [count=\j from 0] in {00,01,10,11}{
    \pgfmathsetmacro{\xx}{\xo+(\j+.5)*\cs}
    \node at (\xx,5.30) {\lab};
  }
  \foreach \lab [count=\i from 0] in {00,01,10,11}{
    \pgfmathsetmacro{\yy}{\yt-(\i+.5)*\cs}
    \node[anchor=east] at (1.25,\yy) {\lab};
  }
  \foreach \i in {0,...,3}{
    \foreach \j in {0,...,3}{
      \pgfmathsetmacro{\xx}{\xo+\j*\cs}
      \pgfmathsetmacro{\yy}{\yt-(\i+1)*\cs}
      \pgfmathtruncatemacro{\forb}{(\i==0&&\j<2)||(\i==2&&\j<1)}
      \pgfmathtruncatemacro{\completion}{(\i==0&&\j>=2)||(\i==1)||(\i==3)}
      \pgfmathtruncatemacro{\groups}{\i==2&&\j>=2}
      \pgfmathtruncatemacro{\residual}{\i==2&&\j==1}
      \ifnum\forb=1
        \path[ffcell,ffforbid] (\xx,\yy) rectangle ++(\cs,\cs);
      \else\ifnum\completion=1
        \path[ffcell,ffcompletion] (\xx,\yy) rectangle ++(\cs,\cs);
      \else\ifnum\groups=1
        \path[ffcell,ffgroups] (\xx,\yy) rectangle ++(\cs,\cs);
      \else\ifnum\residual=1
        \path[ffcell,ffresidual] (\xx,\yy) rectangle ++(\cs,\cs);
      \else
        \path[ffcell] (\xx,\yy) rectangle ++(\cs,\cs);
      \fi\fi\fi\fi
    }
  }
  \draw[very thick] (\xo,\yt-4*\cs) rectangle ++(4*\cs,4*\cs);
  \node[align=center] at (3.10,1.15)
    {\(\mathcal S_{\rm c}:\ \mathcal L_{00}[2,4)\ \dot\cup\ \mathcal L_{01}\ \dot\cup\ \mathcal L_{11}\)};
  \node[align=center] at (3.10,.80)
    {\(\mathcal S_g=\mathcal L_{10}[2,4),\quad \mathcal R=\{1001\}\)};
\end{scope}

\begin{scope}[xshift=7.10cm]
  \node[font=\small] at (3.45,7.05)
    {\(\mathrm{(VI)}\quad(\lambda_1,\lambda_2,\lambda_3)=(4,5,7)\)};
  \node at (3.45,6.65)
    {\(\mu_1=1,\ d=8,\ g=4,\ m=2,\ y_0=00\)};
  \def\xo{1.05}\def\yt{5.18}\def\cs{.57}
  \node[rotate=90] at (-.20,2.90) {\ffsparserow};
  \foreach \lab [count=\j from 0] in {000,001,010,011,100,101,110,111}{
    \pgfmathsetmacro{\xx}{\xo+(\j+.5)*\cs}
    \node[rotate=58,anchor=west] at (\xx,5.23) {\lab};
  }
  \foreach \lab [count=\i from 0] in {000,001,010,011,100,101,110,111}{
    \pgfmathsetmacro{\yy}{\yt-(\i+.5)*\cs}
    \node[anchor=east] at (.96,\yy) {\lab};
  }
  \foreach \i in {0,...,7}{
    \foreach \j in {0,...,7}{
      \pgfmathsetmacro{\xx}{\xo+\j*\cs}
      \pgfmathsetmacro{\yy}{\yt-(\i+1)*\cs}
      \pgfmathtruncatemacro{\support}{mod(\i,2)==floor(\j/4)}
      \ifnum\support=0
        \path[ffcell,ffzero] (\xx,\yy) rectangle ++(\cs,\cs);
        \node[text=black!35] at (\xx+.5*\cs,\yy+.5*\cs) {\(\times\)};
      \else
        \pgfmathtruncatemacro{\y}{floor(\i/2)}
        \pgfmathtruncatemacro{\forb}{(\y==0&&\j<2)||(\y==2&&\j<1)}
        \pgfmathtruncatemacro{\completion}{(\y==0&&\j>=2)||(\y==1)||(\y==3)}
        \pgfmathtruncatemacro{\groups}{\y==2&&\j>=4}
        \pgfmathtruncatemacro{\residual}{\y==2&&\j>=1&&\j<4}
        \ifnum\forb=1
          \path[ffcell,ffforbid] (\xx,\yy) rectangle ++(\cs,\cs);
        \else\ifnum\completion=1
          \path[ffcell,ffcompletion] (\xx,\yy) rectangle ++(\cs,\cs);
        \else\ifnum\groups=1
          \path[ffcell,ffgroups] (\xx,\yy) rectangle ++(\cs,\cs);
        \else\ifnum\residual=1
          \path[ffcell,ffresidual] (\xx,\yy) rectangle ++(\cs,\cs);
        \else
          \path[ffcell] (\xx,\yy) rectangle ++(\cs,\cs);
        \fi\fi\fi\fi
      \fi
    }
  }
  \draw[very thick] (\xo,\yt-8*\cs) rectangle ++(8*\cs,8*\cs);
  \node[align=center] at (3.35,.42)
    {\(\mathcal S_{\rm c}:\ \mathcal L_{00}[2,8)\ \dot\cup\ \mathcal L_{01}\ \dot\cup\ \mathcal L_{11}\)};
  \node[align=center] at (3.35,.07)
    {\(\mathcal S_g=\mathcal L_{10}[4,8),\quad
      \mathcal R=\mathcal L_{10}[1,4)\)};
\end{scope}

\begin{scope}[yshift=-.63cm,xshift=.38cm]
  \path[ffcell,ffforbid] (0,0) rectangle ++(.34,.34);
  \node[anchor=west] at (.45,.17) {\ffforbidden};
  \path[ffcell,ffcompletion] (3.05,0) rectangle ++(.34,.34);
  \node[anchor=west] at (3.50,.17) {\ffcompletionlegend};
  \path[ffcell,ffgroups] (7.15,0) rectangle ++(.34,.34);
  \node[anchor=west] at (7.60,.17) {\ffgroupslegend};
  \path[ffcell,ffresidual] (11.05,0) rectangle ++(.34,.34);
  \node[anchor=west] at (11.50,.17) {\ffresiduallegend};
\end{scope}
\end{tikzpicture}}
\caption{Binary instances of Regimes III and VI with $\mu_1=1$. Gray marks inadmissible words; diagonal hatching, the tail of completion-capable layer $y_0=00$ and full admissible layers; crosshatching, complete $g$-groups in the other partial layer; dots, the residual. The left panel uses column layers as rows, the right panel matrix rows; crosses lie outside the support.}
\label{fig:staircase-geometry}
\ifdefined\TITReview
\end{figure}
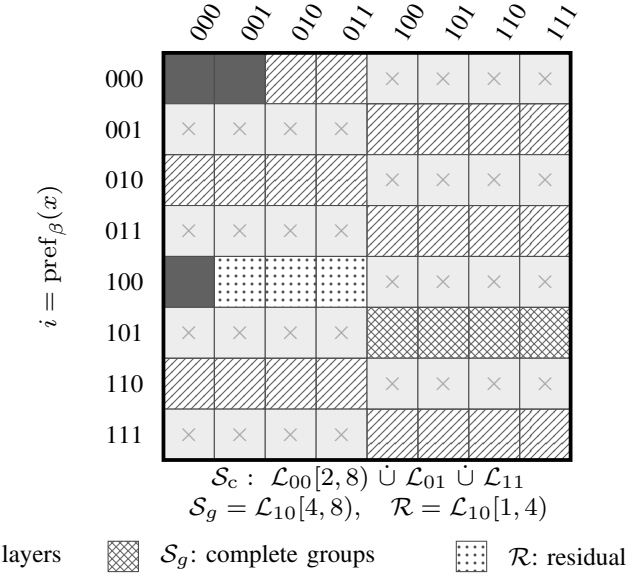
\else
\end{figure*}
\fi
\endgroup

\paragraph*{Zero remainder modulo \(d\)}
If \(R=0\), then \(\nu=m=0\) and every threshold is \(0\) or \(d\). Take all admissible words in reverse order of complete \(g\)-groups. The column group counts and selected row-block masses are nondecreasing, so the covariance argument below gives feasible group endpoints; Lemma~\ref{lem:column-layer-interpolation} fills each group. Thus the construction reaches \(\Omega_2\).

Henceforth \(R>0\). For $y>\ell$, the layer is empty, partial, or full according as
$y\bmod A<k$, $=k$, or $>k$; the partial layer starts at $R$.
The layer $\ell$ starts at $\max(m,[\mu_1-(\ell\bmod A)d]_0^d)$.

Write $f=A-k-1$ for the full-layer count in a row block whose layer indices all exceed $\ell$.

\begin{proposition}[Full-layer completion]
\label{prop:aligned-completion}
Suppose $\nu>0$ and a full admissible layer exists. Choose its smallest
index $y_0=b_0A+s_0$. Select $\mathcal L_{y_0}[m,d)$ and then all
other full layers. Next process the partial layers with index at least $b_0A$ in decreasing layer order, taking their complete $g$-groups in decreasing column order.
Every resulting group endpoint is feasible. If an admissible layer
remains at an index below $b_0A$, the additional packet
\begin{equation}
 \mathcal L_{y_0}[0,m)\ \dot\cup\ \mathcal L_\ell[m,d)
 \label{eq:last-completion-packet}
\end{equation}
is interpolable. The remaining words are the unfilled completion head
and the short boundary fragments of the partial layers.
\end{proposition}
\begin{proof}
\emph{Reverse progression.} The first stage follows Theorem~\ref{thm:complete-layer-core}.
Put $h=g\lceil R/g\rceil$ and $\kappa=\ell+|\mathcal Y_0|$, the constant column value after the full-layer stage. Let $V$ count the selected words by row. At each subsequent group endpoint,
\[
 a=\kappa\e+w,\quad c_{\rm current}=c^\circ-V,\quad
 w=\sum_{j=0}^{B-1}n_j\mathbf1_{[jg,(j+1)g)},
\]
where $n_j$ is nondecreasing. Set $Z_b=V([bg,(b+1)g))$.
For $b>b_0$, $Z_b=fd+t_b$, with $0\le t_b\le d-h$;
the reverse order makes $t_b$ nondecreasing. Earlier blocks have mass
zero. If $s_0=k+1$, the initial block has mass at most
$fd-m+d-h$; its partial groups are processed after all later blocks.
If $s_0\ge k+2$, it has at most $f-1$ full layers and one boundary
layer, so its mass is at most $fd-m$. Thus $Z_b$ is nondecreasing.

The $g$-block masses of $c^\circ$ are equal. Chebyshev's inequality gives
\[
 w^{\mathsf T}c^\circ=\frac{|w|\,|c^\circ|}{d},\qquad
 w^{\mathsf T}V\ge\frac{|w|\,|V|}{d}.
\]
The profile covariance criterion proves feasibility, and each group
is interpolable by Lemma~\ref{lem:column-layer-interpolation}.

\emph{Row-block transition.} Minimality of $y_0$ implies that the only possible remaining layer with index below $b_0A$
is $\ell=b_0A-1$, with threshold $m$. For this case put
$p=B-b_0\ge1$. After \eqref{eq:last-completion-packet},
\[
 \begin{gathered}
 a^+=(\kappa+1)\e+p\mathbf1_{[h,d)},\\
 Z^+=(0,\ldots,0,d-m,Z_*,\ldots,Z_*),
 \end{gathered}
\]
where $\kappa=\ell+pf$, $Z_*=(f+1)d-h\ge d$, and there are $p$
copies of $Z_*$. If $h=sg<d$, the covariance margin is
\[
 J=\frac pB\sum_{i<s}\sum_{j\ge s}(Z_j^+-Z_i^+).
\]
The profile sums give $q^{\lambda_3}/4-\Phi\ge J$. For $b_0\ge2$, $J\ge pZ_*/B\ge pg\ge1$; for $b_0=1$,
$J\ge p(B-s)m/B=p(B-s)\nu\ge1$.
If $h=d$, then $\ell=(k+1)B-1$, $b_0A=(k+1)B$, and
\[
 \kappa+1=H-p(k+1)
 =H\left(1-\frac{b_0(B-b_0)}{B^2}\right)\ge3H/4.
\]
The constant-profile square margin is at least $H^2d/16\ge1$.
The packet has distinct column indices and terminal integral slack,
so Theorem~\ref{thm:matrix-interpolation} applies.
Formula~\eqref{eq:local-layer-threshold} gives the stated residual.
\end{proof}

\paragraph*{No fully admissible layer}
No full layer exists exactly when $\mu_1>(A-1)d$. Since
$\mu_1\le3Ad/4$, this forces $A\in\{2,3\}$.
Here $\ell\ge(A-1)B$. For every remaining complete $g$-group $G$,
\[
 \sum_{w\in G}((c_{\rm current})_{j(w)}-a_{i(w)})
 \le Q-\mu_1-g\ell=(H-2\ell)g-\nu\le0.
\]
The interaction term further decreases the increment. After all such groups, nothing remains if $\nu=0$. For $\nu>0$, put $h=(c+1)g$. The Kraft bound gives $c\le B-2$,
and
\[
 \ell-A(c+1)=(A-1)(B-c)-A\ge0.
\]
All residual rows are therefore at least $h$. Ordinary fragments have
columns in $[R,h)$; the sole possible exception lies in layer $\ell$.
Consequently, every residual cycle would lie in that single layer,
and every column in the row--column intersection occurs there only
once. Single-layer acyclicity and the terminal bound
\eqref{eq:terminal-endpoint} complete the interpolation.

\paragraph*{Zero remainder modulo \(g\)}
If $\nu=0$ and a full layer exists, every threshold is a multiple of $g$. Select the full layers
first and then the partial layers in reverse group order. The full-layer
endpoints have constant prefix profile. In subsequent endpoints the
column block counts and row block masses are nondecreasing: later row
blocks have mass $fd+t_b$, while the first active block has at most this
mass. The preceding covariance argument applies and no short fragments
remain.

\subsubsection{Terminal residual}
After Proposition~\ref{prop:aligned-completion}, with $\nu>0$, treat an empty completion head first. Otherwise remove any extra boundary fragment, then compare the row-block index $b$ with the column-block index $c$. Put $h=(c+1)g$ and $r=h-R=g-\nu$, the ordinary fragment length. Let $c^*=C_3^*\e$ be the terminal residual row profile, and let $u_{\rm remaining}$ count the unselected admissible words by row. At every intermediate
state,
\begin{equation}
 \begin{gathered}
 c_{\rm current}=c^*+u_{\rm remaining},\\
 \operatorname{supp}(c^*)\subseteq
 [0,\lfloor(\mu_1B-1)/H\rfloor].
 \end{gathered}
 \label{eq:fragment-residual-profile}
\end{equation}

If the completion head is empty, the remaining fragments are $[R,h)$,
one per active row block. Let the first such block have index $b$.
For $c<b$, residual row and column indices are disjoint. For $c\ge b$,
$c^*$ vanishes on $[R,h)$ and the residual mass there is at most $r$.
Since $a\ge1$, each fragment has nonpositive total increment.
Thus the fragments can be removed one layer at a time.

Suppose the completion head $E=\mathcal L_{y_0}[0,m)$ is nonempty.
Write $b=\lfloor\ell/A\rfloor$ and $\tau=\ell\bmod A$.
Every later block has the fragment $\mathcal L_{b'A+k}[R,h)$.
In block $b$, the fragment starts at $R$ if $\tau<k$, at $\max(m,R)$
if $\tau=k$, and at $m$ if $\tau>k$.

\paragraph*{The extra boundary fragment}
If $\tau>k$, or $\tau=k$ and $m>h$, put $h_m=g\lceil m/g\rceil$ and add
\[
 E\ \dot\cup\ \mathcal L_\ell[m,h_m).
\]
Its columns traverse $[0,h_m)$ once. With $p=B-b-1$ and
$\kappa'=H-p(k+1)$, the new prefix profile is
$\kappa'\e+p\mathbf1_{[h,d)}$. Later row blocks have mass
$Z_*=(f+1)d-h$; block $b$ has mass $(A-\tau)d-m$ if $\tau>k$, and
$(f+1)d-m$ if $\tau=k$. These masses are nondecreasing.
For $p>0,h<d$, the covariance margin is
\[
 \frac pB\sum_{i<h/g}\sum_{j\ge h/g}(Z_j-Z_i).
\]
It is at least $pg$ when $b\ge1$, and at least
$p(B-h/g)\nu$ when $b=0,\tau>k$.
If $b=0,\tau=k$, then $k=c=0$ and the margin is
$(B-1)^2(m-g)/B\ge1$; here $m-g$ is a positive multiple of $q$.
If $p=0$ or $h=d$, the constant coefficient is at least $3H/4$,
giving square margin at least $H^2d/16$.
The packet is therefore interpolable with one unit of terminal slack.
Only ordinary fragments remain, as treated above.

\paragraph*{The completion head and ordinary fragments}
We may now assume $\tau\le k$, and $m\le h$ when $\tau=k$.
Set $p=B-b$ and $\kappa=\ell+pf$. The current prefix profile is
$\kappa\e+p\mathbf1_{[h,d)}$.
If $c<b$, the ordinary fragments have no outgoing residual arcs;
all possible cycles would lie in $E$. The entire residual is thus
interpolable. If $c>b$, its mass on $[R,h)$ is at most $r$.
Remove ordinary fragments in decreasing layer index, then interpolate $E$.

For $c=b$, first consider $p\ge3$, $c\ge1$; the other cases are handled below. We have
\begin{equation}
 c(A-1)+\tau=kB,\qquad p=B-c.
 \label{eq:aligned-fragment-relation}
\end{equation}
After $E$ is selected, the prefix profile is
$\kappa\e+pv_0+v_1$, where $v_0=\mathbf1_{[h,d)}$,
$v_1=\mathbf1_{[0,m)}$. The first $c$ row blocks of the selected
profile $V$ have mass zero; its last $p$ blocks have equal mass
$Z_*=(f+1)d-h$. Put $N=|V|$, $m_0=(p-1)g$.
The initial profile is nondecreasing in each $g$-block, so its
centered product with $v_1$ is nonpositive. Expansion gives
\begin{equation}
 \Phi\le\frac{|a|\,|c_{\rm current}|}{d}
 +\frac{(pm_0+m)N}{d}-(pv_0+v_1)^{\mathsf T}V.
 \label{eq:fragment-centered-bound}
\end{equation}
The last two terms are nonpositive.
Indeed, $(pv_0+v_1)^{\mathsf T}V\ge(p-1)N$ and
$(pm_0+m)/d\le p-1$ when $m<h$. For $m\ge h$, write
$m/g=c+j+u$, with $j\ge1$, $0\le u<1$. Equal active block masses give
\[
 (pv_0+v_1)^{\mathsf T}V\ge(p-1+j/p)N.
\]
The corresponding coefficient bound follows from
\[
 \frac{pm_0+m}{d}-(p-1)
 =\frac{j+u-c(p-2)}{c+p}\le\frac jp,
\]
since $pu\le c[j+p(p-2)]$. This certifies the head endpoint.

For the subsequent lower fragments, the mass on $[R,h)$ is at most
$r$. To see this, in the first active row block the contribution of
$c^*$ and the remaining boundary fragment is contained, in local word
coordinates, in $[kd+R,kd+h)$. All other fragments lie in different
row blocks. Each lower fragment therefore has nonpositive increment.
The last fragment is in one layer and interpolates to the full endpoint.

\paragraph*{Small numbers of active blocks}
Continue with $c=b$. Let $C=|c_{\rm current}|$ be the residual mass before adding the head. Counting the active layers gives
\[
 \frac Cd=\frac{Acp+p-c\tau}{B}.
\]
The profile sums give \(|F_2|+C=Hd(1-K(F_1))\), so
\[
 C\le Hd/4\ \Longrightarrow\
 |F_2|\ge Hd\bigl(3/4-K(F_1)\bigr)\ge\mu_2,
\]
so the current chain already covers every required cardinality. For $c=0$, \eqref{eq:aligned-fragment-relation}
gives $k=\tau=0$ and $C=d\le Hd/4$.
For $p=1$ and $p=2$, respectively,
\begin{align*}
 4BC/d-AB^2&=-A(B-2)^2+4-4(B-1)\tau,\\
 4BC/d-AB^2&=-A(B-4)^2+8-4(B-2)\tau.
\end{align*}
The first is nonpositive. For $B=2,3$, this follows from
$\tau\equiv1\pmod2$ and $\tau\equiv2\pmod3$, respectively; larger $B$
follow directly. The second is nonpositive for $B\ge6$.
For $B=4,5$, use $\tau\equiv2\pmod4$ and $\tau\equiv3\pmod5$;
the possibility $A=2,B=4$ is excluded by $\tau<A$.
The remaining case is $A=B=3$, $c=1$, $k=\tau=1$, for which $C=7g$.
Here $m=3\nu\le2g$, so one lower fragment has length
$g-\nu\ge g/3$. After selecting it,
$C'\le20g/3<27g/4=Hd/4$, and $\Phi\le HC'\le H^2d/4$.
This single-layer transition covers all remaining required cardinalities.

\begin{theorem}[Completion in Regimes III and VI]
\label{thm:third-terminal-configuration}
In Regimes III and VI, every prescribed middle-layer cardinality
under the Kraft bound is attained by a nested feasible chain.
\end{theorem}
\begin{proof}
The zero-remainder cases and the no-full-layer construction are complete. Otherwise Proposition~\ref{prop:aligned-completion} reaches
the terminal residual. The preceding cases exhaust this residual,
with single-layer, disjoint-index, or integral-slack interpolation at
each transition. The chain ends at $\Omega_2$ or at a cardinality
already sufficient for every allowed $\mu_2$.
\end{proof}

\subsection{Regime V: phase-aligned complete \(g\)-groups}
\label{subsec:regime-v}

Regime V is characterized by
\begin{equation}
  \alpha<\lambda_1\leq\beta,
  \qquad \alpha<\beta.
  \label{eq:regv-range}
\end{equation}
The shortest-layer shadow crosses the common substring coordinate.
We first identify endpoints with a common feasibility certificate, then
construct entry paths and transitions between them. Subcritical paths
usually join a standard chain; a binary boundary path enters the terminal
residual directly. Above the subcritical range, only the critical
odd-alphabet branch needs a separate support construction.

\subsubsection{Phase geometry and the endpoint certificate}

The column periods and their row images have the following scales:
\begin{equation}
\begin{gathered}
 H=q^\alpha,\qquad
 g=q^{\lambda_1-\alpha},\qquad Q=Hg=q^{\lambda_1},\\
 d=q^\beta,\qquad R=d/H,\qquad M=d/Q=R/g.
\end{gathered}
\label{eq:regv-scales}
\end{equation}
There are \(H\) physical layers, each mapping into \(R\) row indices.
A source \(Q\)-block maps into a row \(g\)-block, and each layer has
\(M\) source blocks. To locate the incomplete column group, write
\begin{equation}
 \mu_1=\ell g+\nu,\qquad 0\leq\nu<g,
 \label{eq:regv-mu1-division}
\end{equation}
and define
\begin{equation}
\begin{aligned}
 c&=\left\lceil\frac{\mu_1}{g}\right\rceil,
 &h&=H-c,
 &x&=c/H,\\
 \rho&=h/H,
 &m_1&=\nu d/g,
 &e_1&=\mathbf{1}_{[0,m_1)},\\
 a_0&=(A_3^\circ)^{\mathsf T}\e=\ell\e+e_1,\\
 P&=\e^{\mathsf T}a_0=\ell d+m_1,\\
 K_1&:=K(F_1)=\frac{\mu_1}{q^{\lambda_1}}
              =\frac{P}{q^{\lambda_2}}.
\end{aligned}
\label{eq:regv-phase-parameters}
\end{equation}
Thus \(h\) complete groups remain in each period, occupying fraction
\(\rho\); \(m_1\) is the width of the elevated baseline head.

Intervals use numerical column indices. For \(S\subseteq[0,d)\) and
physical layer \(0\leq y<H\), let
\begin{equation}
  S_y=\{yz\in\Omega_2:[z]_q\in S\}.
  \label{eq:regv-layer-notation}
\end{equation}

The complete admissible column groups form
\begin{equation}
 \mathcal K=\{0\leq z<d:z\bmod Q\geq cg\}.
 \label{eq:regv-K}
\end{equation}

For \(F\subseteq\Omega_2\), abbreviate
\begin{equation}
 u_F=\mathcal Q(F)\e,\qquad v_F=\mathcal Q(F)^{\mathsf T}\e,
 \qquad N_F=|F|,
 \label{eq:regv-uvn}
\end{equation}
and put \(c_0=C_3^\circ\e\).  The natural-order shortest layer gives
the periodic profile
\begin{equation}
 c_0(j)=w(j\bmod g),\qquad
 w(r)=H-[\mu_1-rH]_0^H,
 \label{eq:regv-c0-profile}
\end{equation}
where \([t]_0^H=\min\{H,\max\{0,t\}\}\).  Thus \(w\) is
nondecreasing and \(w(r)=H\) for
\(r\geq\lceil\mu_1/H\rceil\).

Put \(\eta=m_1\bmod g\). An endpoint is \emph{phase-aligned}
if it is a disjoint union of complete column \(g\)-groups and at most
one tail with column interval \([ag+\eta,(a+1)g)\) in one physical
layer, for an integer \(a\geq0\). Column contributions from different
layers are counted with multiplicity.

The phase slack measures the deviation of the baseline overlap from its mean-profile value. Define
\begin{align}
 \Delta_0&=\frac{P(Hd-P)}d-a_0^{\mathsf T}c_0,\nonumber\\
 \Delta_v(F)&=N_F\frac{Hd-P}{d}-v_F^{\mathsf T}c_0,\nonumber\\
 \Xi(F)&=\Delta_0+\Delta_v(F).
 \label{eq:regv-phase-slack}
\end{align}

\begin{lemma}[Exact phase-slack identity]
\label{lem:regv-slack-identity}
For every \(F\subseteq\Omega_2\),
\begin{align}
 \frac{q^{\lambda_3}}4-\Phi(F)
={}&\frac{(P+N_F-Hd/2)^2}{d}
 +(e_1+v_F)^{\mathsf T}u_F \nonumber\\
 &-\frac{N_F(m_1+N_F)}d+\Xi(F).
 \label{eq:regv-slack-identity}
\end{align}
For every phase-aligned endpoint, \(\Xi(F)\geq0\); for
the standard endpoints in \eqref{eq:regv-standard-endpoints}, equality
holds.  
\end{lemma}

\begin{proof}
Substitute \(A_3(F)^{\mathsf T}\e=a_0+v_F\), \(C_3(F)\e=c_0-u_F\) in \eqref{eq:Phi} and complete the square in
\(P+N_F\) to obtain \eqref{eq:regv-slack-identity}.
Set \(\bar w=H-\mu_1/g\) and
\(\psi(t)=t\bar w-\sum_{r=0}^{t-1}w(r)\), \(0\leq t\leq g\).
The profile is nondecreasing with mean \(\bar w\), so
\(\psi(t)\geq0\) and \(\psi(g)=0\). Periodicity gives
\(\Delta_0=\psi(\eta)\). A complete group contributes zero to
\(\Delta_v\), and the tail \([ag+\eta,(a+1)g)\) contributes
\(-\psi(\eta)\). Thus \(\Xi\) is either \(\psi(\eta)\) or zero.
The standard endpoints have precisely this tail when \(\eta>0\),
so their phase slack vanishes.
\end{proof}

\subsubsection{Subcritical standard chain}

Assume first that
\begin{equation}
 c\leq H/2,\qquad \rho=1-c/H\geq1/2.
 \label{eq:regv-subcritical}
\end{equation}

If \(m_1=0\), select \(\mathcal K_y\) for \(y=H-1,\ldots,c\).
After \(j\) layers, \(N=j\rho d\) and \(v=j\mathbf1_{\mathcal K}\).
The row profile has equal mass in successive \(g\)-blocks on
\([(H-j)R,d)\). Since \(\mathcal K\) retains the last \(h\) blocks
of each \(Q\)-period, \(\mathbf1_{\mathcal K}^{\mathsf T}u\geq\rho N\).
Hence \(v^{\mathsf T}u\geq N^2/d\), and
\eqref{eq:regv-slack-identity} certifies the endpoints.
Lemma~\ref{lem:column-layer-interpolation} fills each layer; the
result is the complete admissible core \(T\) of
\eqref{eq:terminal-first-set}.

For the rest of the subcritical construction, assume \(m_1>0\).
Write \(m_1=bQ+s\), \(0\leq s<Q\), and set
\begin{equation}
 (U,V)=
 \begin{cases}
 (m_1,(b+1)Q),&s>cg,\\
 (bQ+cg,bQ+cg+s),&s\leq cg,
 \end{cases}
 \label{eq:regv-UV}
\end{equation}
\begin{equation}
 \mathcal L=\mathcal K\cap[U,d),\qquad
 \mathcal B=\mathcal K\cap[V,d).
 \label{eq:regv-LB}
\end{equation}
Thus \(\mathcal L=\mathcal K\cap[m_1,d)\) and
\(\mathcal B\subseteq\mathcal L\). Directly
substituting the natural-order forbidden intervals in the sparse support
criterion shows that every word in the layer sets used below belongs to
\(\Omega_2\).

For \(0\leq p\leq h-1\), define
\begin{equation}
 G_p=
 \left(\mathop{\dot\bigcup}_{r=0}^{p-1}\mathcal K_{H-1-r}\right)
 \mathbin{\dot\cup}\mathcal L_{H-1-p}
 \mathbin{\dot\cup}\mathcal B_{H-2-p}.
\label{eq:regv-standard-endpoints}
\end{equation}

The set \(G_p\) consists of \(p\) complete \(\mathcal K\)-layers and two boundary
layers of types \(\mathcal L\) and \(\mathcal B\). Empty boundary sets are omitted. Except for the binary direct entry in
\eqref{eq:binary-direct-entry}, the subcritical branches enter at
\(G_0=\mathcal L_{H-1}\mathbin{\dot\cup}\mathcal B_{H-2}\)
and follow
\begin{equation}
 \varnothing\rightsquigarrow G_0\rightsquigarrow G_1
 \rightsquigarrow\cdots\rightsquigarrow F_{\rm V}^{\star}
 \rightsquigarrow F^{(1)}\rightsquigarrow\Omega_2.
 \label{eq:regv-construction-chain}
\end{equation}
Standard promotion ends at \(G_{h-1}\), with boundary layers \(c,c-1\);
the terminal residual uses \(F^{(1)}\) in \eqref{eq:terminal-first-set}.
Fig.~\ref{fig:regv-phase-geometry} illustrates the two phase modes.

\ifdefined\TITReview
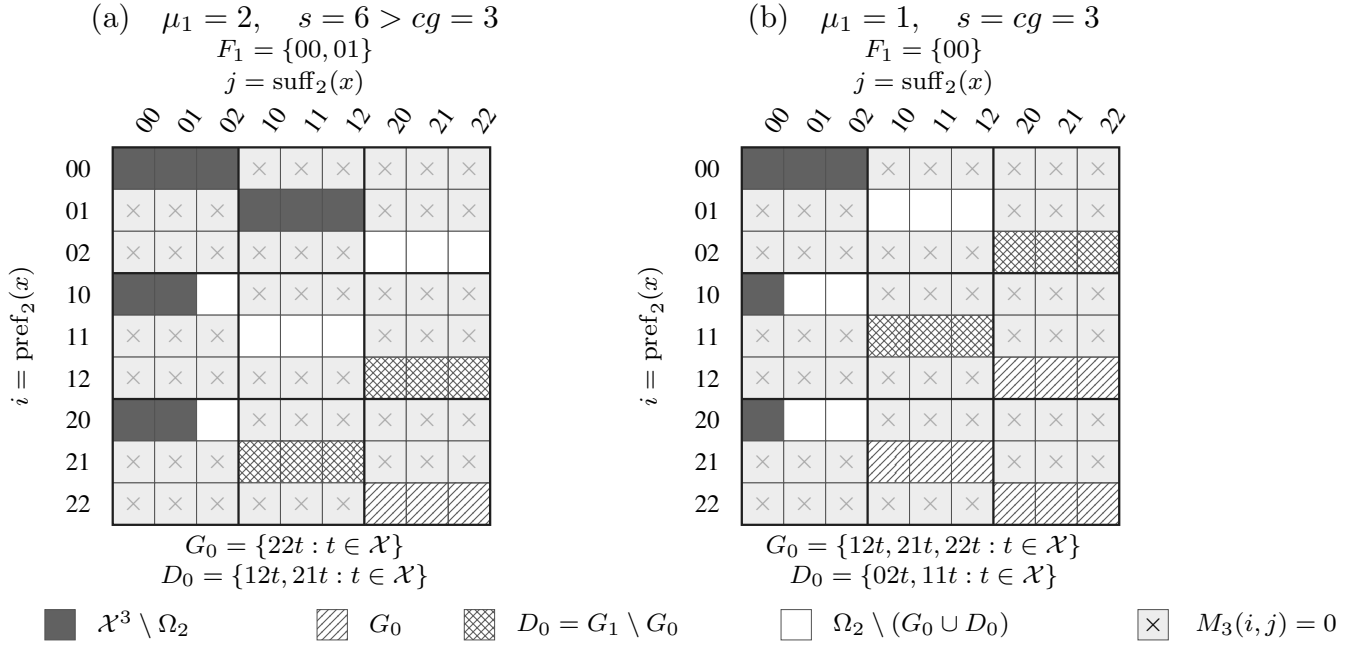
\begin{figure}[!ht]
\centering
\resizebox{\textwidth}{!}{\begin{tikzpicture}[font=\scriptsize]
\tikzset{
  rvcell/.style={draw=black!68,line width=.32pt},
  rvforbid/.style={fill=black!62},
  rvbase/.style={pattern=north east lines,pattern color=black!64},
  rvnext/.style={pattern=crosshatch,pattern color=black!62},
  rvopen/.style={fill=white},
  rvzero/.style={fill=black!7},
  rvblock/.style={draw=black!88,line width=.72pt}
}
\def\rvcs{.47}
\def\rvxo{1.38}
\def\rvyt{4.42}


\begin{scope}[xshift=0cm]
  \path (0,-.72) rectangle (6.82,5.96);
  \node[font=\small] at (3.42,5.86)
    {\(\mathrm{(a)}\quad \mu_1=2,\quad s=6>cg=3\)};
  \node at (3.42,5.50) {\(F_1=\{00,01\}\)};
  \node at (3.42,5.14) {\(j=\suff_2(x)\)};
  \node[rotate=90] at (.36,2.31) {\(i=\pref_2(x)\)};

  \foreach \lab [count=\j from 0] in
    {00,01,02,10,11,12,20,21,22}{
    \pgfmathsetmacro{\xx}{\rvxo+(\j+.5)*\rvcs}
    \node[rotate=55,anchor=west] at (\xx,4.50) {\lab};
  }
  \foreach \lab [count=\i from 0] in
    {00,01,02,10,11,12,20,21,22}{
    \pgfmathsetmacro{\yy}{\rvyt-(\i+.5)*\rvcs}
    \node[anchor=east] at (1.27,\yy) {\lab};
  }

  \foreach \i in {0,...,8}{
    \foreach \j in {0,...,8}{
      \pgfmathsetmacro{\xx}{\rvxo+\j*\rvcs}
      \pgfmathsetmacro{\yy}{\rvyt-(\i+1)*\rvcs}
      \pgfmathtruncatemacro{\support}{mod(\i,3)==floor(\j/3)}
      \ifnum\support=0
        \path[rvcell,rvzero] (\xx,\yy) rectangle ++(\rvcs,\rvcs);
        \node[text=black!34] at ({\xx+.5*\rvcs},{\yy+.5*\rvcs}) {\(\times\)};
      \else
        \pgfmathtruncatemacro{\forb}{(\i<2)||(\j<2)}
        \pgfmathtruncatemacro{\gzero}{(floor(\i/3)==2)&&(\j>=6)}
        \pgfmathtruncatemacro{\dzero}{
          ((floor(\i/3)==2)&&(\j>=3)&&(\j<6))||
          ((floor(\i/3)==1)&&(\j>=6))}
        \ifnum\forb=1
          \path[rvcell,rvforbid] (\xx,\yy) rectangle ++(\rvcs,\rvcs);
        \else\ifnum\gzero=1
          \path[rvcell,rvbase] (\xx,\yy) rectangle ++(\rvcs,\rvcs);
        \else\ifnum\dzero=1
          \path[rvcell,rvnext] (\xx,\yy) rectangle ++(\rvcs,\rvcs);
        \else
          \path[rvcell,rvopen] (\xx,\yy) rectangle ++(\rvcs,\rvcs);
        \fi\fi\fi
      \fi
    }
  }
  \draw[rvblock] (\rvxo,\rvyt-9*\rvcs) rectangle ++(9*\rvcs,9*\rvcs);
  \foreach \k in {3,6}{
    \draw[rvblock]
      ({\rvxo+\k*\rvcs},\rvyt-9*\rvcs) -- ({\rvxo+\k*\rvcs},\rvyt);
    \draw[rvblock]
      (\rvxo,{\rvyt-\k*\rvcs}) -- ({\rvxo+9*\rvcs},{\rvyt-\k*\rvcs});
  }
  \node at (3.42,-.06) {\(G_0=\{22t:t\in\X\}\)};
  \node at (3.42,-.39) {\(D_0=\{12t,21t:t\in\X\}\)};
\end{scope}

\begin{scope}[xshift=7.05cm]
  \path (0,-.72) rectangle (6.82,5.96);
  \node[font=\small] at (3.42,5.86)
    {\(\mathrm{(b)}\quad \mu_1=1,\quad s=cg=3\)};
  \node at (3.42,5.50) {\(F_1=\{00\}\)};
  \node at (3.42,5.14) {\(j=\suff_2(x)\)};
  \node[rotate=90] at (.36,2.31) {\(i=\pref_2(x)\)};

  \foreach \lab [count=\j from 0] in
    {00,01,02,10,11,12,20,21,22}{
    \pgfmathsetmacro{\xx}{\rvxo+(\j+.5)*\rvcs}
    \node[rotate=55,anchor=west] at (\xx,4.50) {\lab};
  }
  \foreach \lab [count=\i from 0] in
    {00,01,02,10,11,12,20,21,22}{
    \pgfmathsetmacro{\yy}{\rvyt-(\i+.5)*\rvcs}
    \node[anchor=east] at (1.27,\yy) {\lab};
  }

  \foreach \i in {0,...,8}{
    \foreach \j in {0,...,8}{
      \pgfmathsetmacro{\xx}{\rvxo+\j*\rvcs}
      \pgfmathsetmacro{\yy}{\rvyt-(\i+1)*\rvcs}
      \pgfmathtruncatemacro{\support}{mod(\i,3)==floor(\j/3)}
      \ifnum\support=0
        \path[rvcell,rvzero] (\xx,\yy) rectangle ++(\rvcs,\rvcs);
        \node[text=black!34] at ({\xx+.5*\rvcs},{\yy+.5*\rvcs}) {\(\times\)};
      \else
        \pgfmathtruncatemacro{\forb}{(\i<1)||(\j<1)}
        \pgfmathtruncatemacro{\gzero}{
          ((floor(\i/3)==2)&&(\j>=3))||
          ((floor(\i/3)==1)&&(\j>=6))}
        \pgfmathtruncatemacro{\dzero}{
          ((floor(\i/3)==1)&&(\j>=3)&&(\j<6))||
          ((floor(\i/3)==0)&&(\j>=6))}
        \ifnum\forb=1
          \path[rvcell,rvforbid] (\xx,\yy) rectangle ++(\rvcs,\rvcs);
        \else\ifnum\gzero=1
          \path[rvcell,rvbase] (\xx,\yy) rectangle ++(\rvcs,\rvcs);
        \else\ifnum\dzero=1
          \path[rvcell,rvnext] (\xx,\yy) rectangle ++(\rvcs,\rvcs);
        \else
          \path[rvcell,rvopen] (\xx,\yy) rectangle ++(\rvcs,\rvcs);
        \fi\fi\fi
      \fi
    }
  }
  \draw[rvblock] (\rvxo,\rvyt-9*\rvcs) rectangle ++(9*\rvcs,9*\rvcs);
  \foreach \k in {3,6}{
    \draw[rvblock]
      ({\rvxo+\k*\rvcs},\rvyt-9*\rvcs) -- ({\rvxo+\k*\rvcs},\rvyt);
    \draw[rvblock]
      (\rvxo,{\rvyt-\k*\rvcs}) -- ({\rvxo+9*\rvcs},{\rvyt-\k*\rvcs});
  }
  \node at (3.42,-.06) {\(G_0=\{12t,21t,22t:t\in\X\}\)};
  \node at (3.42,-.39) {\(D_0=\{02t,11t:t\in\X\}\)};
\end{scope}

\begin{scope}[xshift=.62cm,yshift=-1.10cm]
  \path[rvcell,rvforbid] (0,0) rectangle ++(.34,.34);
  \node[anchor=west] at (.45,.17) {\(\X^3\setminus\Omega_2\)};
  \path[rvcell,rvbase] (3.05,0) rectangle ++(.34,.34);
  \node[anchor=west] at (3.50,.17) {\(G_0\)};
  \path[rvcell,rvnext] (4.70,0) rectangle ++(.34,.34);
  \node[anchor=west] at (5.15,.17) {\(D_0=G_1\setminus G_0\)};
  \path[rvcell,rvopen] (8.25,0) rectangle ++(.34,.34);
  \node[anchor=west] at (8.70,.17)
    {\(\Omega_2\setminus(G_0\cup D_0)\)};
  \path[rvcell,rvzero] (12.25,0) rectangle ++(.34,.34);
  \node at (12.42,.17) {\(\times\)};
  \node[anchor=west] at (12.76,.17) {\(M_3(i,j)=0\)};
\end{scope}
\end{tikzpicture}}
\caption{Ternary phase modes for \((\lambda_1,\lambda_2,\lambda_3)=(2,3,4)\),
\(\alpha=1,\beta=2,H=g=3,Q=d=9\). Cell \((i_1i_2,i_2j_2)\)
represents \(i_1i_2j_2\); crosses are ambient zeros. Gray denotes words
forbidden by \(F_1\), diagonal hatching \(G_0\), and crosshatching
\(D_0=G_1\setminus G_0\). Panels (a), (b) show \(s>cg\), \(s\leq cg\).}
\label{fig:regv-phase-geometry}
\end{figure}
\else
\begin{figure*}[!t]
\centering
\resizebox{\textwidth}{!}{\begin{tikzpicture}[font=\scriptsize]
\tikzset{
  rvcell/.style={draw=black!68,line width=.32pt},
  rvforbid/.style={fill=black!62},
  rvbase/.style={pattern=north east lines,pattern color=black!64},
  rvnext/.style={pattern=crosshatch,pattern color=black!62},
  rvopen/.style={fill=white},
  rvzero/.style={fill=black!7},
  rvblock/.style={draw=black!88,line width=.72pt}
}
\def\rvcs{.47}
\def\rvxo{1.38}
\def\rvyt{4.42}


\begin{scope}[xshift=0cm]
  \path (0,-.72) rectangle (6.82,5.96);
  \node[font=\small] at (3.42,5.86)
    {\(\mathrm{(a)}\quad \mu_1=2,\quad s=6>cg=3\)};
  \node at (3.42,5.50) {\(F_1=\{00,01\}\)};
  \node at (3.42,5.14) {\(j=\suff_2(x)\)};
  \node[rotate=90] at (.36,2.31) {\(i=\pref_2(x)\)};

  \foreach \lab [count=\j from 0] in
    {00,01,02,10,11,12,20,21,22}{
    \pgfmathsetmacro{\xx}{\rvxo+(\j+.5)*\rvcs}
    \node[rotate=55,anchor=west] at (\xx,4.50) {\lab};
  }
  \foreach \lab [count=\i from 0] in
    {00,01,02,10,11,12,20,21,22}{
    \pgfmathsetmacro{\yy}{\rvyt-(\i+.5)*\rvcs}
    \node[anchor=east] at (1.27,\yy) {\lab};
  }

  \foreach \i in {0,...,8}{
    \foreach \j in {0,...,8}{
      \pgfmathsetmacro{\xx}{\rvxo+\j*\rvcs}
      \pgfmathsetmacro{\yy}{\rvyt-(\i+1)*\rvcs}
      \pgfmathtruncatemacro{\support}{mod(\i,3)==floor(\j/3)}
      \ifnum\support=0
        \path[rvcell,rvzero] (\xx,\yy) rectangle ++(\rvcs,\rvcs);
        \node[text=black!34] at ({\xx+.5*\rvcs},{\yy+.5*\rvcs}) {\(\times\)};
      \else
        \pgfmathtruncatemacro{\forb}{(\i<2)||(\j<2)}
        \pgfmathtruncatemacro{\gzero}{(floor(\i/3)==2)&&(\j>=6)}
        \pgfmathtruncatemacro{\dzero}{
          ((floor(\i/3)==2)&&(\j>=3)&&(\j<6))||
          ((floor(\i/3)==1)&&(\j>=6))}
        \ifnum\forb=1
          \path[rvcell,rvforbid] (\xx,\yy) rectangle ++(\rvcs,\rvcs);
        \else\ifnum\gzero=1
          \path[rvcell,rvbase] (\xx,\yy) rectangle ++(\rvcs,\rvcs);
        \else\ifnum\dzero=1
          \path[rvcell,rvnext] (\xx,\yy) rectangle ++(\rvcs,\rvcs);
        \else
          \path[rvcell,rvopen] (\xx,\yy) rectangle ++(\rvcs,\rvcs);
        \fi\fi\fi
      \fi
    }
  }
  \draw[rvblock] (\rvxo,\rvyt-9*\rvcs) rectangle ++(9*\rvcs,9*\rvcs);
  \foreach \k in {3,6}{
    \draw[rvblock]
      ({\rvxo+\k*\rvcs},\rvyt-9*\rvcs) -- ({\rvxo+\k*\rvcs},\rvyt);
    \draw[rvblock]
      (\rvxo,{\rvyt-\k*\rvcs}) -- ({\rvxo+9*\rvcs},{\rvyt-\k*\rvcs});
  }
  \node at (3.42,-.06) {\(G_0=\{22t:t\in\X\}\)};
  \node at (3.42,-.39) {\(D_0=\{12t,21t:t\in\X\}\)};
\end{scope}

\begin{scope}[xshift=7.05cm]
  \path (0,-.72) rectangle (6.82,5.96);
  \node[font=\small] at (3.42,5.86)
    {\(\mathrm{(b)}\quad \mu_1=1,\quad s=cg=3\)};
  \node at (3.42,5.50) {\(F_1=\{00\}\)};
  \node at (3.42,5.14) {\(j=\suff_2(x)\)};
  \node[rotate=90] at (.36,2.31) {\(i=\pref_2(x)\)};

  \foreach \lab [count=\j from 0] in
    {00,01,02,10,11,12,20,21,22}{
    \pgfmathsetmacro{\xx}{\rvxo+(\j+.5)*\rvcs}
    \node[rotate=55,anchor=west] at (\xx,4.50) {\lab};
  }
  \foreach \lab [count=\i from 0] in
    {00,01,02,10,11,12,20,21,22}{
    \pgfmathsetmacro{\yy}{\rvyt-(\i+.5)*\rvcs}
    \node[anchor=east] at (1.27,\yy) {\lab};
  }

  \foreach \i in {0,...,8}{
    \foreach \j in {0,...,8}{
      \pgfmathsetmacro{\xx}{\rvxo+\j*\rvcs}
      \pgfmathsetmacro{\yy}{\rvyt-(\i+1)*\rvcs}
      \pgfmathtruncatemacro{\support}{mod(\i,3)==floor(\j/3)}
      \ifnum\support=0
        \path[rvcell,rvzero] (\xx,\yy) rectangle ++(\rvcs,\rvcs);
        \node[text=black!34] at ({\xx+.5*\rvcs},{\yy+.5*\rvcs}) {\(\times\)};
      \else
        \pgfmathtruncatemacro{\forb}{(\i<1)||(\j<1)}
        \pgfmathtruncatemacro{\gzero}{
          ((floor(\i/3)==2)&&(\j>=3))||
          ((floor(\i/3)==1)&&(\j>=6))}
        \pgfmathtruncatemacro{\dzero}{
          ((floor(\i/3)==1)&&(\j>=3)&&(\j<6))||
          ((floor(\i/3)==0)&&(\j>=6))}
        \ifnum\forb=1
          \path[rvcell,rvforbid] (\xx,\yy) rectangle ++(\rvcs,\rvcs);
        \else\ifnum\gzero=1
          \path[rvcell,rvbase] (\xx,\yy) rectangle ++(\rvcs,\rvcs);
        \else\ifnum\dzero=1
          \path[rvcell,rvnext] (\xx,\yy) rectangle ++(\rvcs,\rvcs);
        \else
          \path[rvcell,rvopen] (\xx,\yy) rectangle ++(\rvcs,\rvcs);
        \fi\fi\fi
      \fi
    }
  }
  \draw[rvblock] (\rvxo,\rvyt-9*\rvcs) rectangle ++(9*\rvcs,9*\rvcs);
  \foreach \k in {3,6}{
    \draw[rvblock]
      ({\rvxo+\k*\rvcs},\rvyt-9*\rvcs) -- ({\rvxo+\k*\rvcs},\rvyt);
    \draw[rvblock]
      (\rvxo,{\rvyt-\k*\rvcs}) -- ({\rvxo+9*\rvcs},{\rvyt-\k*\rvcs});
  }
  \node at (3.42,-.06) {\(G_0=\{12t,21t,22t:t\in\X\}\)};
  \node at (3.42,-.39) {\(D_0=\{02t,11t:t\in\X\}\)};
\end{scope}

\begin{scope}[xshift=.62cm,yshift=-1.10cm]
  \path[rvcell,rvforbid] (0,0) rectangle ++(.34,.34);
  \node[anchor=west] at (.45,.17) {\(\X^3\setminus\Omega_2\)};
  \path[rvcell,rvbase] (3.05,0) rectangle ++(.34,.34);
  \node[anchor=west] at (3.50,.17) {\(G_0\)};
  \path[rvcell,rvnext] (4.70,0) rectangle ++(.34,.34);
  \node[anchor=west] at (5.15,.17) {\(D_0=G_1\setminus G_0\)};
  \path[rvcell,rvopen] (8.25,0) rectangle ++(.34,.34);
  \node[anchor=west] at (8.70,.17)
    {\(\Omega_2\setminus(G_0\cup D_0)\)};
  \path[rvcell,rvzero] (12.25,0) rectangle ++(.34,.34);
  \node at (12.42,.17) {\(\times\)};
  \node[anchor=west] at (12.76,.17) {\(M_3(i,j)=0\)};
\end{scope}
\end{tikzpicture}}
\caption{Ternary phase modes for \((\lambda_1,\lambda_2,\lambda_3)=(2,3,4)\),
\(\alpha=1,\beta=2,H=g=3,Q=d=9\). Cell \((i_1i_2,i_2j_2)\)
represents \(i_1i_2j_2\); crosses are ambient zeros. Gray denotes words
forbidden by \(F_1\), diagonal hatching \(G_0\), and crosshatching
\(D_0=G_1\setminus G_0\). Panels (a), (b) show \(s>cg\), \(s\leq cg\).}
\label{fig:regv-phase-geometry}
\end{figure*}
\fi

\begin{proposition}[Common subcritical endpoint]
\label{prop:regv-common-endpoint}
Every endpoint \(G_p\) in \eqref{eq:regv-standard-endpoints} satisfies
\begin{equation}
  \Phi(G_p)\leq q^{\lambda_3}/4.
  \label{eq:regv-common-endpoint-bound}
\end{equation}
\end{proposition}

\begin{proof}
Let \(u_p,v_p,N_p\) denote the profiles and cardinality of \(G_p\), and
put
\begin{equation}
 \tau=
 \begin{cases}
 (b+1)Q,&s>cg,\\ bQ,&s\leq cg.
 \end{cases}
 \label{eq:regv-tau}
\end{equation}
Split the row mass as
\begin{equation}
 N_1=u_p([0,\tau)),\qquad N_2=u_p([\tau,d)),
 \qquad N_p=N_1+N_2.
 \label{eq:regv-N12}
\end{equation}
The two boundary layers give the exact count
\begin{equation}
 \frac{N_p+m_1}{d}
 =p\rho+\frac{\tau}{d}+2\rho\frac{d-\tau}{d}.
 \label{eq:regv-endpoint-count}
\end{equation}
The row images of \(\mathcal B\subseteq\mathcal L\subseteq\mathcal K\)
form a two-dimensional tail array.  Its monotonicity gives the prefix
bound
\begin{equation}
 \frac{N_1}{N_p}\leq\frac{\tau}{d}.
 \label{eq:regv-macro-prefix}
\end{equation}
The row-count compensation in
Appendix~\ref{app:regv-phase-counts} gives
\begin{equation}
 (e_1+v_p)^{\mathsf T}u_p
 \geq p\rho N_p+N_1+2\rho N_2.
 \label{eq:regv-inner-common}
\end{equation}
Since \(1-2\rho\leq0\), equations
\eqref{eq:regv-endpoint-count}--\eqref{eq:regv-inner-common} imply
\begin{align}
 \frac{(e_1+v_p)^{\mathsf T}u_p}{N_p}
 &\geq p\rho+2\rho+(1-2\rho)\frac{N_1}{N_p}\nonumber\\
 &\geq p\rho+2\rho+(1-2\rho)\frac{\tau}{d}
 =\frac{N_p+m_1}{d}.
 \label{eq:regv-inner-target}
\end{align}
Lemma~\ref{lem:regv-slack-identity} gives \(\Xi(G_p)=0\); retaining its
nonnegative square term proves \eqref{eq:regv-common-endpoint-bound}.
\end{proof}

\paragraph*{Standard promotion}

For adjacent standard endpoints, direct subtraction in
\eqref{eq:regv-standard-endpoints} gives
\begin{equation}
 D_p=G_{p+1}\setminus G_p
 = (\mathcal K\setminus\mathcal L)_{H-1-p}
 \mathbin{\dot\cup}(\mathcal L\setminus\mathcal B)_{H-2-p}
 \mathbin{\dot\cup}\mathcal B_{H-3-p}.
 \label{eq:regv-promotion}
\end{equation}
The three column projections are pairwise disjoint and partition
\(\mathcal K\); hence \(D_p\) satisfies column-side uniqueness.

\begin{proposition}[Standard subcritical promotion]
\label{prop:regv-subcritical-promotion}
Every standard promotion $G_p\subseteq G_{p+1}$ covers all
intermediate cardinalities.
\end{proposition}
\begin{proof}
The difference has column-side uniqueness.
Appendix~\ref{app:standard-integral-slack} supplies one unit of
integral slack at \(G_{p+1}\).
Apply Theorem~\ref{thm:matrix-interpolation}.
\end{proof}

\subsubsection{Subcritical entry bridges}

The common certificate proves that $G_0$ is feasible; an entry path
must still reach it from the empty set. The position of the initial
support relative to $e_1$ determines this path. The lower-phase branch is
\begin{align}
 m_1&\leq \mu_1d/Q,\nonumber\\
 \frac{m_1}{d}&\leq\frac{c-1}{H-1}\qquad(m_1>0),
 \label{eq:regv-lower-phase}
\end{align}
where the second inequality is equivalent to the first whenever \(m_1>0\).

\begin{proposition}[Lower-phase entry bridge]
\label{prop:regv-lower-bridge}
Under \eqref{eq:regv-lower-phase}, every cardinality between the empty
set and \(G_0\) is feasible.
\end{proposition}

\begin{proof}
Use
\begin{equation}
 \varnothing\longrightarrow\mathcal L_{H-1}
 \longrightarrow G_0.
 \label{eq:regv-lower-bridge-chain}
\end{equation}
At the single-layer endpoint, \(R<Q\) puts every row image in \(\mathcal L\),
so \(v^{\mathsf T}u=N\) and \(m_1+N\leq d\). For \(Q\leq R\), the periodic tail count
\begin{equation}
 \Theta_H(n)=(H-c)\left\lfloor\frac nH\right\rfloor
 +\min\{H-c,n\bmod H\}\geq\rho n
 \label{eq:regv-periodic-tail}
\end{equation}
controls the unique truncated fiber. Substitution in
\eqref{eq:regv-slack-identity}, retaining its square and phase slack,
proves feasibility. Proposition~\ref{prop:regv-common-endpoint}
certifies \(G_0\); both differences in
\eqref{eq:regv-lower-bridge-chain} are single-layer and interpolate by
Lemma~\ref{lem:column-layer-interpolation}.
\end{proof}

The upper-phase branch is
\begin{equation}
 m_1>\mu_1d/Q;\qquad\text{in particular, }m_1>0\text{ and }\ell=c-1.
 \label{eq:regv-upper-phase}
\end{equation}
The preceding two-step bridge applies outside the balanced case
\(x=1/2\), \(H\mid M\). Appendix~\ref{app:regv-binary-bridge}
treats \(H=2\); the remaining balanced branch is
\begin{equation}
 x=1/2,\qquad H\geq4,\qquad Q=Hg,\qquad H\mid M.
 \label{eq:regv-balanced-long}
\end{equation}

\begin{proposition}[Upper-phase entry]
\label{prop:regv-upper-bridge}
Under \eqref{eq:regv-upper-phase}, a nested feasible chain reaches
$G_0$, or $F^{(1)}$ in the binary boundary case
\eqref{eq:binary-direct-entry}.
\end{proposition}

\begin{proof}
To reduce to \eqref{eq:regv-balanced-long}, put \(F=\mathcal L_{H-1}\) and let \(f(z)=d-R+\lfloor z/H\rfloor\)
be this layer's row map. The common power scales \(H\), \(M\) give
either \(M<H\) or \(H\mid M\).

For \(M<H\), we have \(M\leq H/2\); the cutoffs in \eqref{eq:regv-UV} give
\[
 f(\mathcal L)\subseteq\mathcal L,
 \qquad f(\mathcal L)\cap[0,m_1)=\varnothing.
\]
Thus \(N=|\mathcal L|\) satisfies \((e_1+v_F)^{\mathsf T}u_F=N\). Also \(\mathcal L\subseteq[m_1,d)\), hence \(m_1+N\leq d\).
Nonnegative phase slack and Lemma~\ref{lem:regv-slack-identity} yield
\[
 \frac{q^{\lambda_3}}4-\Phi(F)
 \geq
 \frac{(P+N-Hd/2)^2}{d}
 +\frac{N(d-m_1-N)}d\geq0.
\]

For \(H\mid M\) and \(x<1/2\), put
\[
 \begin{aligned}
 \delta&=H/2-c\geq1/2, & \sigma&=s/Q,\\
 n&=M-b, & \varepsilon&=\Theta_H(n)-\rho n\geq0,
 \end{aligned}
\]
with \(\Theta_H\) as in \eqref{eq:regv-periodic-tail}. Let \(\mathscr D(F)\) be
\eqref{eq:regv-slack-identity} without \(\Xi(F)\).
For \(\sigma\leq x\), set \(z=xn-\sigma\). Periodic tail counting gives
\[
 \frac{M}{Q}\mathscr D(F)
  \geq (\delta M+z)^2-\rho(M-n)z+\rho M(\varepsilon-\sigma)
 \geq\frac{M^2-M}{4}.
\]
since \(z\geq0\), \(2\delta M-\rho(M-n)\geq0\), and \(\rho\sigma\leq\rho x\leq1/4\). For \(\sigma>x\), put
\(t=\sigma-x\), \(z=x(n-1)\), \(\iota_H=\mathbf{1}_{\{1\leq n\bmod H\leq H-c\}}\). The unique truncated source block gives
\[
 \begin{aligned}
 \frac{M}{Q}\mathscr D(F)
  \geq{}&(\delta M+z)^2-\rho(M-n)z+\rho M(\varepsilon-x)\\
       &+t\{M(1-\iota_H)-z\}
 \geq\frac{M^2-2M}{4}.
 \end{aligned}
\]
Here the first two terms total at least \(M^2/4\); the third loses at most
\(M/4\). The last is nonnegative for \(\iota_H=0\) and loses at most \(M/4\)
for \(\iota_H=1\). Both bounds are nonnegative because \(M\geq H\geq3\).
Since \(\Xi(F)\geq0\), the single-layer endpoint \(F\) is feasible.

The two-step bridge therefore applies unless \(x=1/2\) and \(H\mid M\):
Proposition~\ref{prop:regv-common-endpoint} certifies \(G_0\), and each
difference is single-layer. Appendix~\ref{app:regv-binary-bridge}
treats \(H=2\). For \(H\geq4\), the remaining branch is
\eqref{eq:regv-balanced-long}.

Lemma~\ref{lem:regv-balanced-entry} supplies this bridge.
\end{proof}

\begin{lemma}[Balanced entry bridge]
\label{lem:regv-balanced-entry}
Under \eqref{eq:regv-upper-phase} and \eqref{eq:regv-balanced-long},
a nested feasible chain joins the empty set to \(G_0\).
\end{lemma}
\begin{proof}
\emph{Four-step construction.}
In \eqref{eq:regv-balanced-long}, put $D=H^2Q$, $A=d-m_1$, and
$p_0=HQ/4$. The lower and upper layers have indices $H-1$ and $H-2$.
Define the periodic lower tails
\begin{equation}
 T=\{j\in\mathcal L:j\bmod(HQ)\ge HQ/2\}.
 \label{eq:balanced-periodic-tails}
\end{equation}
For the $k$th block, counted from the right, the two upper packets are
\begin{align}
 P_k&=\mathcal K\cap[d-(k-\tfrac12)HQ,d-(k-1)HQ),\nonumber\\
 W_k&=\mathcal K\cap[d-kHQ,d-(k-\tfrac12)HQ).
 \label{eq:balanced-upper-packets}
\end{align}
Each has $p_0$ words. A completed block adds its periodic lower tails,
the lower copy of $W_k$, and the two upper packets. In the current block,
the four steps are
\begin{equation}
 \begin{gathered}
 \text{lower tails in the right half}\ \longrightarrow\ (P_k)_{H-2}\\
 \longrightarrow\ \text{remaining lower tails and }(W_k)_{H-1}\\
 \longrightarrow\ (W_k)_{H-2}.
 \end{gathered}
 \label{eq:regv-four-step}
\end{equation}
All lower tails are intersected with $\mathcal L$.
Write $F_{k-1,r}$ for the endpoint after step $r$ in block $k$,
with $F_{k-1,0}$ the preceding endpoint. Thus
\begin{equation}
 F_{k-1,0}\subseteq\cdots\subseteq F_{k-1,4}=F_{k,0}.
 \label{eq:regv-four-step-endpoints}
\end{equation}
At contact in the right half, $A\le(k-1/2)D$, stop after step 2;
otherwise execute steps 3 and 4. Every packet difference lies in one
physical layer.

\emph{Endpoint bounds.}
The row map $f_l(j)=d-R+\lfloor j/H\rfloor$
sends each lower periodic tail to the $\mathcal K$ part of its
corresponding $HQ$ interval. Its lower copy supplies one unit of
support, and the upper copy supplies a second unit once selected.
For $k\ge2$, put $t_{\rm prev}=(k-1)Hp_0$. If the first step adds
$u$ lower tails, then
\begin{align*}
 I_1-N_1&\ge(k-1)(H-3)p_0,\\
 I_2-N_2&\ge[(H-3)(k-1)-1]p_0.
\end{align*}
When step 3 is reached, the first-step tails are complete and their
cumulative mass is $t_{\rm old}=(k-1/2)Hp_0$.
If $t\ge t_{\rm old}$ is the new tail mass, then
\begin{align*}
 I_3-N_3&\ge[k(H-3)+1-H/2]p_0,\\
 I_4-N_4&\ge[k(H-3)-H/2]p_0.
\end{align*}
All four bounds are nonnegative for $H\ge4,k\ge2$.

For a complete first block, $A>D$, put
$T_1=T\cap[d-D,d)$ and $T_0=T\cap[d-D/2,d)$.
Their masses are $t=Hp_0$ and $t_0=Hp_0/2$. The four endpoints are
\[
 (T_0,\varnothing),\quad(T_0,P_1),\quad
 (T_1\cup W_1,P_1),\quad(T_1\cup W_1,P_1\cup W_1).
\]
Their masses are $t_0,t_0+p_0,t+2p_0,t+3p_0$, and the lower
tails give overlap at least $t_0,2t_0,t+t_0,2t$, respectively.
Thus $I_r\ge N_r$: at the last two endpoints the differences are
at least $(H/2-2)p_0$ and $(H-3)p_0$. Continue with the second block.

For a first-block head contact, $D/2<A\le D$, put
$T_0=T\cap[d-D/2,d)$, $t_0=|T_0|=D/8$, and $t=|T|$.
The four pairs of lower and upper column sets are
\[
 (T_0,\varnothing),\quad(T_0,P_1),\quad
 (T\cup W_1,P_1),\quad(T\cup W_1,P_1\cup W_1).
\]
Here $R=d/H$ is a multiple of $HQ$. The upper $P_1$ has row image
in $e_1$ or in $T$, so it contributes at least $p_0$ at the last
two endpoints. Consequently,
\begin{align*}
 N_1&=t_0,& I_1&\ge t_0,\\
 N_2&=t_0+p_0,&I_2&\ge2t_0,\\
 N_3&=t+2p_0,&I_3&\ge t+t_0+p_0,\\
 N_4&=t+3p_0,&I_4&\ge2t+p_0.
\end{align*}
Since $t\ge t_0=(H/2)p_0\ge2p_0$, all satisfy $I_r\ge N_r$.

For $0<A\le D/2$, use $(T,\varnothing)$ followed by $(T,P_1)$ and
then complete the lower layer. When $A\ge HQ/2$, $P_1$ is fully
admissible and $f_l(T)\subseteq\mathcal L\cap P_1$.
Every complete $HQ$ interval has equally many head and tail words,
and a right truncation favors the tail. Hence
$|T|\ge|\mathcal L\setminus T|$ and the lower-completion endpoint has
\[
 I\ge2|T|+|P_1|\ge N.
\]
If $A<HQ/2$, then $T=\mathcal L$ and $f_l(\mathcal L)\subseteq\mathcal L$.
Select the lower layer directly and then interpolate to $G_0$.

\emph{Lower completion.}
Fix the upper packets after contact and select all of $\mathcal L$.
Write $n_1=|\mathcal L|$, $n_2$ for the upper mass, $N=n_1+n_2$,
$r=A-N$, and $\delta=\min(s,Q-s)$. Then
\begin{equation}
 n_1=(A+\delta)/2,\qquad r=(A-\delta)/2-n_2.
 \label{eq:balanced-completion-count}
\end{equation}
At a head contact with $k\ge2$, $n_2=kHQ/2$ and
\[
 r\ge(k-1)D/2+D/4-kHQ/2,\qquad I\ge N-n_2/2.
\]
If $A>R$, the condition $2rA\ge dn_2$ follows from
\[
 2r-Hn_2\ge(HQ/2)[(k-1)H-2k]\ge0.
\]
At a tail contact with $k\ge2$, $n_2=(2k-1)HQ/4$ and
$I\ge N-C$, where $C=(k-1)HQ/4$. In this case
\[
 r-HC\ge(HQ/4)[(k-1)(H-2)-1]>0,
\]
so $Ar\ge dC$ when $A>R$. For $A\le R$, all upper row images
lie in $e_1$ and $I\ge N$.
These inequalities certify the completion endpoint through
\eqref{eq:regv-slack-identity}.

For a first-block head contact, the upper set is
$S_0=P_1\cup W_1=\mathcal K\cap[d-HQ,d)$.
Put $E=2|T|-n_1\ge0$. The lower contribution is $n_1+E$,
and the upper contribution is at least $n_2-C$, with $0\le C\le HQ/4$.
It suffices to prove
\begin{equation}
 Ar/d+E\ge C.\label{eq:first-head-completion}
\end{equation}
For $A\le R$, $C=0$. For $A>R$ and $A\ge D/2+HQ$,
$r\ge D/4$ proves the claim. In the remaining interval
$D/2\le A<D/2+HQ$, the common $q$-power scales give either
$R\le D/4$ or $R=D/2$.
In the former case,
\[
 Ar/d\ge HQ/2-Q\ge HQ/4.
\]
In the latter, write $A=R+t$, $0<t<HQ$; the baseline bound is
$Ar/d\ge HQ/4-Q/2$. If $t\le Q/2$, the unsupported half of the
upper row image lies in $e_1$, so $C=0$.
For $Q/2\le t\le HQ-Q$, periodic tail counting gives $E\ge Q/2$.
For $t\ge HQ-Q$,
\[
 Ar/d\ge\frac{(R-Q)(R+HQ-Q)}{2HR}\ge HQ/4;
\]
the last inequality is $R(H-2)\ge(H-1)Q$, valid for $R=H^2Q/2$.
Thus \eqref{eq:first-head-completion} holds throughout.

All selected upper packets are complete $g$-groups, and the only
partial lower group is complementary to the cutoff $m_1$.
Thus $\Xi\ge0$ at every macro-endpoint. After lower completion,
select the remaining upper words to reach the common endpoint $G_0$.
Each transition is single-layer, so
Theorem~\ref{thm:matrix-interpolation} supplies every intermediate
cardinality.

\end{proof}

\subsubsection{Supercritical complete-group chains}

For \(c>H/2\), the square term certifies most complete-group endpoints. The critical odd-alphabet case requires additional overlap support.
\paragraph*{Noncritical branch}
Put \(z=P-Hd/2\). For even \(q\), \(c>H/2\) implies \(\ell\geq H/2\) and \(z\geq m_1\);
for odd \(q\) with \(\ell\geq(H+1)/2\), we have \(z\geq d/2+m_1\).
At every complete-group endpoint with \(\Xi(F)\ge0\),
Lemma~\ref{lem:regv-slack-identity} gives
\begin{equation}
 \frac{q^{\lambda_3}}4-\Phi(F)
 \geq\frac{z^2+N_F(2z-m_1)}d>0.
 \label{eq:regv-supercritical-easy}
\end{equation}
Thus tail-first complete \(g\)-groups in layers
\(H-1,\ldots,c\) give feasible endpoints;
their single-layer differences interpolate by
Lemma~\ref{lem:column-layer-interpolation}.
\paragraph*{Critical odd-alphabet branch}
For the remaining branch
\begin{equation}
 \ell=(H-1)/2,\qquad 0<m_1<d
 \label{eq:regv-critical-odd}
\end{equation}
put
\begin{equation}
 h=(H-1)/2,\qquad c=(H+1)/2,\qquad\rho=h/H.
 \label{eq:regv-critical-parameters}
\end{equation}
At a phase-aligned complete-group endpoint it suffices to prove
\begin{equation}
\begin{aligned}
 (e_1+v_F)^{\mathsf T}u_F&\geq f_{N_F}(m_1),\\
 f_N(m_1)&=\frac{N(d-m_1)}d-\frac{(d/2-m_1)^2}{d}\\
 &\leq\frac N2+\frac{N^2}{4d}.
\end{aligned}
 \label{eq:regv-critical-target}
\end{equation}

For \(0\leq j<M\), define the complete tail in the \(j\)th source
\(Q\)-block by
\begin{equation}
 K_j=[jQ+cg,(j+1)Q),\qquad
 \mathcal K=\bigcup_{j=0}^{M-1}K_j.
 \label{eq:regv-critical-Kj}
\end{equation}
For \(2\leq n\leq h\), let
\begin{equation}
 G^{\mathrm c}_n=\bigcup_{r=0}^{n-1}\mathcal K_{H-1-r}.
 \label{eq:regv-critical-Gn}
\end{equation}
Modulo-\(H\) tail counting gives
\begin{equation}
 N_n=n\rho d,\qquad
 v_n=n\mathbf{1}_{\mathcal K},\qquad
 \mathbf{1}_{\mathcal K}^{\mathsf T}u_n\geq\rho N_n.
 \label{eq:regv-critical-count}
\end{equation}
Hence
\begin{equation}
 \frac{q^{\lambda_3}}4-\Phi(G^{\mathrm c}_n)
 \geq N_n\left(\frac{3n\rho}{4}-\frac12\right)\geq0.
 \label{eq:regv-critical-Gn-feasible}
\end{equation}

\begin{proposition}[Critical odd-alphabet construction]
\label{prop:regv-critical-construction}
In the branch \eqref{eq:regv-critical-odd}, every cardinality through the
last complete-group endpoint is feasible.
\end{proposition}

\begin{proof}
For \(H\geq5\), we first reach \(G^{\mathrm c}_2\), after which
complete-layer promotion applies. For \(H=3\), the two-layer construction
continues directly to the terminal residual.

\emph{Entry bridge and promotion for \(H\geq5\).}
Start with the common tail \(K_{M-1}\) in layers \(H-1\) and \(H-2\).
For \(M=s_MH\) and \(0\leq b\leq s_M-2\), group the ordinary lower tails as
\begin{equation}
 \mathcal P_b=\bigcup_{j=bH}^{(b+1)H-1}K_j,\qquad
 \tau(b)=(H-1)s_M+b.
 \label{eq:regv-critical-packets}
\end{equation}
Complete the lower set to $\bigcup_{j=(s_M-1)H}^{M-1}K_j$, keeping the upper set $K_{M-1}$;
then process $b=s_M-2,\ldots,0$ with difference
\begin{equation}
 D_b=(\mathcal P_b)_{H-1}
 \mathbin{\dot\cup}(K_{\tau(b)})_{H-2}.
 \label{eq:regv-critical-packet-difference}
\end{equation}
The column supports are disjoint since
\begin{equation}
 \tau(b)-(b+1)H=(H-1)(s_M-b)-H\geq H-2>0.
 \label{eq:regv-critical-support-separation}
\end{equation}
The lower packet maps equally into \(H\) target \(g\)-blocks;
exactly the last \(h\) meet the support tail. With \(N=N_{\rm low}+N_{\rm up}\),
\begin{equation}
\begin{aligned}
 N_{\rm low}:N_{\rm up}&=H:1,\\
 (e_1+v)^{\mathsf T}u&\geq\frac{H-1}{H+1}N,
 &N/d&<1/2.
\end{aligned}
 \label{eq:regv-critical-packet-inner}
\end{equation}
For \(H\geq5\), the coefficient is at least \(2/3\) and \(1/2+N/(4d)<5/8\),
giving real slack at least \(N/24\). Since \(g\geq q\) and the first packet
has \(N=(H+1)hg\geq60\), every endpoint has more than one unit of slack.
Apply Theorem~\ref{thm:matrix-interpolation} to
\eqref{eq:regv-critical-packet-difference}. For \(M<H\), the seed
supports the entire lower layer and both completions are single-layer.
After the packets, add the missing upper tails within one layer to
reach \(G^{\mathrm c}_2\), then begin the single-layer transitions \(G^{\mathrm c}_2\to\cdots\to G^{\mathrm c}_h\).

\emph{Entry bridge and promotion for \(H=3\).}
Here \(q=3\), \(Q=3g\), and \eqref{eq:regv-critical-target} becomes
\begin{equation}
 I:=(e_1+v)^{\mathsf T}u\geq
 B(N):=\frac N2+\frac{N^2}{4d}.
 \label{eq:regv-H3-target}
\end{equation}
The lower and upper layers are 2 and 1. Put $U=\mathcal K\cap[m_1,d)$. Let $t$ count
the complete $K_j$ in $U$, namely its last $t$ groups.
For $M=3s\ge27$, process $k=0,\ldots,2M/27-1$ until these upper groups are exhausted.
Set $z=M-9k-1$, $p=M-3k-1$. The four additions are
\[
\begin{array}{c|c|l}
 &\text{layer}&\text{group indices}\\ \hline
1&2&z\\
2&1&p\\
3&2&z-3,\ z-6,\ p-1,\ p-2\\
4&1&p-1,\ p-2.
\end{array}
\]
The lower row map sends $z,z-3,z-6$ to $p,p-1,p-2$, respectively.
Step 1 contains the lower copy of $p$; step 3 supplies the other
two supports, whose indices are $1,0$ modulo 3, whereas source tails
have index 2. All groups are distinct and lower indices are at least
$M/3$. For the cumulative endpoints $F_{k,r}$,
\begin{equation}
 \Delta N/g=(1,1,4,2),\quad\Delta I/g\ge(1,1,2,2).
 \label{eq:regv-H3-four-increments}
\end{equation}
Thus the five endpoints \(F_{k,r}\), \(0\leq r\leq4\), have
\begin{equation}
\begin{array}{c@{\quad}ccccc}
 &0&1&2&3&4\\
 N/g&8k&8k+1&8k+2&8k+6&8k+8\\
 I/g\ \ge&6k&6k+1&6k+2&6k+4&6k+6.
\end{array}
\label{eq:regv-H3-four-endpoints}
\end{equation}
The first, second, and fourth satisfy \(I\geq3N/4\) and \(N<d\),
hence \eqref{eq:regv-H3-target}. At the third, the current block
exists only if \(d\geq27(k+1)g\), and
\begin{equation}
\begin{aligned}
&108(k+1)(2k+1)-(8k+6)^2\\
&\hspace{3em}=4(38k^2+57k+18)>0.
\end{aligned}
\label{eq:regv-H3-tight-endpoint}
\end{equation}
which is the required \(I-N/2\geq N^2/(4d)\).

\emph{Truncation at upper exhaustion.}
Stop if step 2 exhausts the complete upper groups. Otherwise let
$r$ be the smaller of 2 and their remaining number. Step 3 takes
lower groups $p-1,p-2$ and $z-3i$, $1\le i\le r$; step 4 takes upper groups $p-i$,
$1\le i\le r$. With $u=rg$, the third endpoint satisfies
\[
 N_3=(8k+4)g+u,\qquad I_3\ge(6k+2)g+u.
\]
Using $d\ge27(k+1)g$,
\begin{align*}
 I_3-N_3/2-N_3^2/(4d)
 &\ge\frac{g(152k^2+174k+18)}{108(k+1)}>0.
\end{align*}
Step 4 adds $rg$ words and at least $rg$ overlap. As $B'(N)<1$ for
$N<d$, its slack cannot decrease. Each transition is single-layer.

\emph{Completion and continuation.}
Suppose first that $M\ge9$ and the $t\le2M/9$ complete upper groups
are exhausted. Their union is $S=\bigcup_{j=M-t}^{M-1}K_j$.
Complete the lower set to $\mathcal K$. The definition of $t$ gives
$m_1>(M-t)Q-g\ge2d/3$, so every upper row image lies in $e_1$.
The complete lower row profile has mass $g$ per period of length $g$
on $[2d/3,d)$; its mass on $\mathcal K$ is $Mg/3$, and $e_1$
covers at least $M-3t-1$ full periods. Hence at $(\mathcal K,S)$,
\[
 N=(M+t)g,\qquad I\ge(4M/3-2t-1)g,
\]
and
\[
 \frac{I-B(N)}g\ge\frac{3M}4-\frac{8t}3-\frac{t^2}{12M}-1
 \ge\frac{149M}{972}-1>0.
\]
Then complete the upper set to $U$, including its boundary tail.

If $M\ge27$ and $t>2M/9$, complete $k=2M/27$ blocks and then
complete the lower set to $L=\mathcal K\cap[d/3,d)$.
There are $2M/3$ lower groups and $2M/9$ upper groups.
Of these, $2M/9$ lower groups have double support and $2M/27$
upper groups have lower support. Therefore
\[
 \begin{gathered}
 N=12kg,\quad I\ge7kg,\quad d=81kg/2,\\
 I-B(N)\ge kg/9>0.
 \end{gathered}
\]
For $M=9$, start with lower $K_8$, upper $K_8$, lower
$K_5\cup K_7\cup K_6$, and upper $K_7$, stopping if the complete
upper groups are exhausted. If $t\le2$, use the preceding completion.
If $t>2$, complete the lower set to $L$; the endpoint has
$N=8g$, $I\ge5g$, and $I-B(N)\ge11g/27$.

After completion to $L$, complete the upper set and continue as
follows, using pairs of lower and upper column sets:
\[
\begin{array}{ll}
m_1\le d/3:&(L,L)\to(\mathcal K,L)\to(\mathcal K,U),\\
m_1>d/3:&(L,U)\to(\mathcal K,U).
\end{array}
\]
Each endpoint has $N\le2d/3$. Its row images lie in $[d/3,d)$,
where every image in $\mathcal K$ has weight at least two: either
both selected layers cover it, or the upper layer complements $e_1$,
since $\mathbf1_{[0,m_1)}+\mathbf1_U=1$ on $\mathcal K$.
A source $Q$-block maps into $\mathcal K$ exactly when its index is
2 modulo 3. Every terminal string of groups, including a partial
first group, places at least one third of its mass on these indices.
Consequently $I\ge2N/3\ge B(N)$. This also proves the endpoint
$(\mathcal K,U)$ following exhaustion of the complete upper groups.

For $M=3$, select lower and upper copies of $K_2$, then of $K_1$,
then of $K_0$, while the upper group is complete and admissible.
Upon upper exhaustion, complete the lower set to $\mathcal K$.
For $t=0,1,2,3$, the resulting masses $N/g$ are $3,4,5,6$,
and the overlap bounds $I/g$ are $3,3,4,4$, respectively.
Each satisfies $I\ge B(N)$. The intervening endpoint bounds are
$(1,1),(2,2),(3,3),(4,4),(5,4),(6,4)$.
Complete the upper set to $U$ as above.
For $M=1$, select lower $K_0$, then upper $U$.
The overlaps are at least $g$ and $2g$, respectively, and satisfy
$I\ge B(N)$.

Complete groups have zero column-phase contribution; the boundary
tail in $U$ complements the prefix $e_1$, so $\Xi\ge0$ throughout
the endpoint chain. All completions are single-layer transitions.
Theorem~\ref{thm:matrix-interpolation} supplies the intervening
cardinalities and reaches the complete tail core $(\mathcal K,U)$.
This proves the proposition.
\end{proof}

\subsubsection{Terminal residual}

The preceding chains select complete groups. We now absorb the boundary remainders into a common first endpoint, then fill the remaining short gaps.
Put $y=K_1$, $\varepsilon=x-y$, and define
\begin{align}
 \Omega_2&=\{P\le w<Hd:w\bmod Q\ge\mu_1\},\nonumber\\
 T&=\{w\in\Omega_2:w\bmod Q\ge cg\},\nonumber\\
 F^{(1)}&=T\cup\{w\in\Omega_2:w<cd\}.
 \label{eq:terminal-first-set}
\end{align}
The incoming endpoint is \(F_{\rm V}^\star=G_{h-1}\) in the
subcritical case with \(m_1>0\), except that
\eqref{eq:binary-direct-entry} uses $F_{\rm V}^\star=\mathcal L_0$; it is \(F_{\rm V}^\star=T\) in the
critical \(H=3\) case. In all other cases it is
\(\bigcup_{y=c}^{H-1}\mathcal K_y\).
Each lies in \(F^{(1)}\); put
\(\mathcal P_1=F^{(1)}\setminus F_{\rm V}^\star\).
For the usual subcritical entry with \(m_1>0\),
\[
 T\setminus G_{h-1}
 = (\mathcal K\setminus\mathcal L)_c
   \mathbin{\dot\cup}(\mathcal L\setminus\mathcal B)_{c-1}.
\]
These column sets are disjoint, and the incomplete groups in
\(F^{(1)}\setminus T\) have columns outside \(\mathcal K\).
For the binary direct entry, $T\setminus\mathcal L_0=\mathcal K_1$,
so the first packet is again column-unique.
In the supercritical case, \(\mathcal P_1\) lies in layer \(c-1\).
If $\nu=0$, all thresholds are aligned and $\Omega_2=T$.
For $\nu>0$, put $\delta=cg-\mu_1=g-\nu$. The residual after the
first packet is exactly
\begin{equation}
 D_1=\{cd\le w<Hd:\mu_1\le w\bmod Q<cg\}.
 \label{eq:terminal-short-gaps}
\end{equation}
In particular, every row index in $D_1$ is at least $dx=cR$.

\begin{lemma}[First-packet endpoint]
\label{lem:first-packet-endpoint}
The set $F^{(1)}$ satisfies $\Phi(F^{(1)})<H^2d/4$ whenever a
nonempty residual first packet is required. This packet is interpolable.
\end{lemma}
\begin{proof}
Write $a,c_1$ for its prefix and residual profiles, and
$e_0=\mathbf1_{\mathcal K}$. Outside $\mathcal K$, the initial prefix
shadow and selected words occupy only the first $c$ physical layers.
Consequently,
\begin{equation}
 a\le c\e+(H-c)e_0.
 \label{eq:regv-residual-row-bound}
\end{equation}
The residual after all of $\Omega_2$ is selected consists of words
$w<P$ with $w\bmod Q\ge\mu_1$. Its mass is at most $P(1-y)$:
in the partial $Q$-block use $(s-\mu_1)_+\le(1-y)s$.
The additional mass in \eqref{eq:terminal-short-gaps} is
$(H-c)M\delta$. Thus
\begin{equation}
 |c_1|\le Hd(x-x^2+xy-y^2).
 \label{eq:regv-profile-sum-closed}
\end{equation}

Let $\widetilde c=C_3(T)\e$. Since $F^{(1)}\supseteq T$,
$0\le c_1\le\widetilde c$. Counting complete $Q$ tails after the
cutoff $P$ gives $|T|\ge(1-x)(Hd-P)$, hence
\begin{equation}
 |\widetilde c|\le Hd\,x(1-y),\qquad
 e_0^{\mathsf T}\widetilde c\le(1-x)|\widetilde c|.
 \label{eq:enlarged-core-profile}
\end{equation}
For the second inequality, decompose $\widetilde c$ into the
$g$-periodic short-gap profile and a prefix-truncated $g$-periodic
tail profile. The former has exactly fraction $1-x$ of its mass on
$e_0$. In each $Q$-block of the latter, the first $c$ periods are full
whenever any of the last $H-c$ periods is present; the tail average
is therefore at most the full-block average.

Equations~\eqref{eq:regv-residual-row-bound}--
\eqref{eq:enlarged-core-profile} yield
\begin{equation}
 \frac{\Phi(F^{(1)})}{H^2d}
 \le P_B=x[(1-x+x^2)(1-x+\varepsilon)-\varepsilon^2].
 \label{eq:terminal-PB}
\end{equation}
When $H\mid M$, the short-gap profile starts at the $Q$-aligned
index $cR$. The same periodic comparison applies to $c_1$ itself and gives
\begin{equation}
 \frac{\Phi(F^{(1)})}{H^2d}
 \le P_A=(1-x+x^2)[x(1-x)+\varepsilon(x-\varepsilon)].
 \label{eq:terminal-PA}
\end{equation}
Here $P_B-P_A=(1-x)^2\varepsilon^2\ge0$.
For $c\ge2$, maximization at $\varepsilon=1/H$ gives the upper bounds
\begin{align*}
 H\text{ odd}:&\quad
 \frac{(H+1)(3H^3+3H^2-7H+1)}{16H^4},\\
 H\text{ even}:&\quad
 \frac{3H^3+6H^2-4H-8}{16H^3}.
\end{align*}
They are below $1/4$ for $H\ge5$; the subcritical values with
$H\ge3$ are also strict. The remaining small-scale cases
are treated in Appendix~\ref{app:regv-qh-gt-d}.

For $(H,c)=(3,2),(4,3)$ with $H\mid M$, first suppose
$\varepsilon>1/H^2$, equivalently $m_1<(H-1)d/H$.
On the first $c-1$ groups of each $Q$-block in $[d-R,d)$,
the prefix profile is $c-1$, one less than
\eqref{eq:regv-residual-row-bound}. Their residual mass is
$(c-1)d\varepsilon/H$. Subtracting it gives
\[
 \Phi(F^{(1)})/(H^2d)\le P_A-(c-1)\varepsilon/H^3.
\]
For $H=3$ its maximum is $187/756$; for $H=4$ its maximum on
$[0,1/H]$ is $63/256$. If $\varepsilon\le1/H^2$, use $P_A$ directly:
its respective maxima are $161/729$ and $767/4096$.
All four bounds are strictly below $1/4$.

For $H=2$, $P_A<1/4$ throughout the aligned range; the only
enlarged-profile exception has $Q=d$ and is covered by the appendix.
For a subcritical first packet, the column indices are unique.
If $q$ is even, strictness supplies integral slack. If $q$ is odd,
$x\le(H-1)/(2H)$, and differentiation in
\eqref{eq:terminal-PB} gives
\begin{align*}
 1/4-P_B
 &>\frac{H^4-6H^3+16H^2-10H+3}{16H^4}\\
 &\ge\frac{3}{4H^3}.
\end{align*}
The last inequality follows from
$(H-3)(H^3-3H^2+7H-1)\ge0$.
Since $d\ge3H$, the real margin exceeds $9/4$.
The integral-slack alternative therefore applies. In the supercritical
branch the packet lies in one physical layer and is acyclic.
\end{proof}

\paragraph*{The remaining short gaps}
Partition \eqref{eq:terminal-short-gaps} into
\[
 \begin{aligned}
 \mathcal P_2&=\{w\in D_1:j(w)\ge dx\},\\
 \mathcal P_3&=\{w\in D_1:j(w)<dx\}.
 \end{aligned}
\]
After $F^{(1)}$, the residual row profile at $i\ge dx$ is
$g$-periodic, of mass $\delta$ per period. Moreover $a\ge c-1$
and $a=H$ on $\mathcal K$.
Process $\mathcal P_2$ by source $Q$-blocks from right to left.
Each such packet has $\delta$ words and positive linear term at most
$\delta$. For $c\ge2$, its negative linear term is at least
$(c-1)\delta$, so the total increment is nonpositive.

For $c=1,H\mid M$, group the source blocks into $HQ$ intervals.
The first $H-1$ blocks encountered have row image in $\mathcal K$;
the last has row image in the head. Relative to the preceding
interval endpoint, an accumulated packet of mass $N$ and supported
mass $N_t$ has increment at most $N-HN_t$.
During its tail phase $N_t=N$; at its head,
$N_t=(H-1)\delta$, $N\le H\delta$, giving a nonpositive increment.
An initial truncated interval is a terminal string of these blocks.
For $c=1,M<H$, the threshold $dx=Mg$ excludes the short gap in
source block zero. Every remaining block has index $1\le k<M$,
and its row-block index $yM+k$ is nonzero modulo $H$. Its row image
lies in $\mathcal K$, so its increment is at most $(1-H)\delta$.
Every constituent packet is in one physical layer and is interpolable.

Finally, $\mathcal P_3$ has row indices at least $dx$ and column
indices below $dx$. The disjoint-index interpolation lemma connects
the current endpoint to the full endpoint $\Omega_2$.

\begin{theorem}[Closure of Regime V]
\label{thm:regv-closure}
For every parameter tuple in Regime V, every prescribed middle-layer
cardinality under the Kraft bound is attained by a nested feasible chain.
\end{theorem}
\begin{proof}
The aligned subcritical construction reaches \(T=\Omega_2\).
For \(m_1>0\), the usual entry bridges and standard promotions,
the binary direct entry, and the supercritical constructions reach
$F_{\rm V}^\star$. Lemma~\ref{lem:first-packet-endpoint} and the two
remaining packets complete the chain. The terminal saturation bound
ensures that it covers the prescribed $\mu_2$.
\end{proof}

\section{Proof of the Main Result}

We combine the six regimes with nonoverlapping extension to prove the theorem and give the selection procedure.

\begin{proof}[Proof of Theorem~\ref{thm:main-three-length}]
The cases \(s=1\) and \(s=2\) follow from the definitions and
Theorem~\ref{thm:natural-order}. For \(s=3\) choose \(F_1=F_{\lambda_1}(\mu_1)\).
If \(\lambda_3\geq2\lambda_2\), that theorem supplies \(\mu_2\) compatible level-\(\lambda_2\) words,
and Proposition~\ref{prop:nonoverlap} completes the extension.
Otherwise \(\lambda_2<\lambda_3<2\lambda_2\), and the matrix alternative together with the position
of \(\lambda_1\) relative to \(\alpha,\beta\) gives Table~\ref{tab:six-regimes}.

Theorems~\ref{thm:reverse-uniform-regimes},
\ref{thm:third-terminal-configuration}, and \ref{thm:regv-closure}
cover all six regimes and produce \(|F_2|=\mu_2\) with \(\Phi(F_2)\leq q^{\lambda_3}/4\).
Corollary~\ref{cor:quarter} supplies \(\mu_3\) admissible level-\(\lambda_3\) words.
Numerical indices determine all choices and tie-breaking.
\end{proof}

\paragraph*{Construction procedure}
For three lengths, the proof gives the following procedure.
\begin{enumerate}
\item Select the natural-order shortest layer
\(F_1=F_{\lambda_1}(\mu_1)\).
\item If \(\lambda_3\geq2\lambda_2\), use
Theorem~\ref{thm:natural-order} to select \(F_2\) and proceed to step 4.
Otherwise determine the regime from Table~\ref{tab:six-regimes}.
\item In Regimes I, II, and IV, select \(\mu_2\) words in the prescribed
reverse order. In Regimes III, V, and VI, follow the endpoint chain
until it reaches the required cardinality; within the packet containing
that cardinality, use Theorem~\ref{thm:matrix-interpolation}.
Break ties by numerical index and stop at \(|F_2|=\mu_2\).
\item Scan length-\(\lambda_3\) words in numerical order and retain
the first \(\mu_3\) having no prefix or suffix in \(F_1\cup F_2\).
\end{enumerate}

\begin{remark}[Direct implementation cost]
\label{rem:complexity}
Let \(N_i=q^{\lambda_i}\), with numerical indices stored in unit-cost
machine words. Maintain the two profiles and scan the remaining packet
to minimize \eqref{eq:current-increment} at each greedy step. At most
\(N_2\) steps, each scanning at most \(N_2\) candidates, cost
\(O(N_2^2)\) operations, including interval generation and profile
updates. Prefix and reversed-prefix tries for
\(F_1\cup F_2\) test each longest word in \(O(\lambda_3)\) time.
The total cost is \(O(N_2^2+N_3\lambda_3)\), with
\(O(N_2+q^\beta+\mu_3)\) words of storage for the candidates,
profiles, tries, and output indices. Packet intervals generate the
candidates without a dense matrix.
\end{remark}

\section{Conclusion}

The construction proves the \(q\)-ary \(3/4\) conjecture for at most three
lengths. The overlap matrix, fixed-cardinality interpolation, and
phase-aligned \(g\)-groups control the coupled profiles deterministically.

For four or more lengths, subsequent extensions require several
coupled profiles. A multilevel interpolation principle preserving
several matrix objectives is a natural next step.

\section{Row-Count Compensation at Standard Endpoints}
\label{app:regv-phase-counts}

We prove the row-overlap bound \eqref{eq:regv-inner-common} by combining
ordinary tail surplus with boundary compensation. Use
\eqref{eq:regv-standard-endpoints}, counting rows with multiplicity in $u=u_p$.

\begin{lemma}[Standard endpoint compensation]
\label{lem:regv-local-phase-counts}
For every admissible standard endpoint $G_p$, put
\[
 J_t=[tQ,(t+1)Q),\quad N_t=u(J_t),\quad
 S_t=u(J_t\cap\mathcal K)-\rho N_t.
\]
Then $S_t\ge0$. For $s\le cg$, the boundary defect $L_b$ below
satisfies $L_b+pS_b\ge0$. Hence \eqref{eq:regv-inner-common} holds.
\end{lemma}

\begin{proof}
Let $E=(e_1+v_p)^{\mathsf T}u-p\rho N_p-N_1-2\rho N_2$.
For $s\le cg$, in the local coordinate $r=i-bQ$ define
\begin{align*}
 \chi(r)&=2\mathbf1_{[cg,Q)}(r)+\mathbf1_{[0,s)}(r)
                         -\mathbf1_{[cg,cg+s)}(r),\\
 L_b&=\sum_{r=0}^{Q-1}u_{bQ+r}\chi(r)-2\rho N_b.
\end{align*}
Expanding $e_1+v_p=e_1+p\mathbf1_{\mathcal K}
+\mathbf1_{\mathcal L}+\mathbf1_{\mathcal B}$ gives
\begin{equation}
 E=p\sum_tS_t+2\sum_{t>b}S_t+
 \begin{cases}0,&s>cg,\\L_b,&s\le cg.\end{cases}
 \label{eq:app-regv-compensation}
\end{equation}

\emph{1) Tail surplus.}
Each column $Q$-block contains $H$ groups of size $g$;
$\mathcal K$ retains the last $h$. Completing the source layers to
copies of $\mathcal K$ gives a $g$-periodic reference row profile
$\widehat u$, of mass $w_0=hg$ per period. With
$\widehat u(U)=\sum_{r\in U}\widehat u_r$ and $W=Hw_0$,
\[
 \widehat u([cg,Q))=\rho W,\qquad
 \sum_{r=0}^{Q-1}\widehat u_r\chi(r)=2\rho W.
\]
The second identity follows from
$\widehat u([0,s))=\widehat u([cg,cg+s))$.
Every reference right tail has $\mathcal K$-fraction at least $\rho$.
It also has $\chi$-average at least $2\rho$: the weighted deviation
$\chi-2\rho$ has total zero and changes sign only from nonpositive
to nonnegative, at $cg+s$. A partial first row is allowed since its
words have equal weight.

If $d\ge HQ$, then $Q\mid R$, and every nonempty row-block profile
is a reference right tail. For $d<HQ$, only a block containing
both boundary layers may fail this property. Thus $S_t\ge0$ outside
these blocks, and $L_b\ge0$ if $J_b$ is a right tail.

\emph{2) Boundary compensation.}
Put $M=d/Q<H$. A physical layer occupies $M$ row $g$-periods.
In a block containing both boundary layers, let $jg$ be the local
start of the $\mathcal L$ layer; $M\le j\le H-M$.
For $s\le cg$, partition the reference words in increasing row order
as $F_D,F_B,F_C,F_A$: $F_B$ is the actual $\mathcal B$ part,
$F_C$ completes $\mathcal L$ to $\mathcal K$, and $F_A$ consists
of the $\mathcal L$ part and later complete source layers. Writing
$n_X=|F_X|$, the layer cutoffs give
\[
 \begin{aligned}
 n_A&=(H-j-b)w_0,& n_B&=(M-b)w_0-s,\\
 n_C&=bw_0,& n_D&=(j-M+b)w_0+s.
 \end{aligned}
\]
In particular, $n_A\ge n_B$ and $n_D\ge n_C$.
The $F_B$ part ends at $jg$ and $F_A$ starts in period $j+b$.
Unless the entire $\mathcal K$-tail is selected or all selected
rows lie in it, direct counting gives
\[
 HS_t\ge
 \begin{cases}
 [h(c-M)+b(h-c)]w_0,&j\le c,\\
 [h(j-M)+b(H-2c)]w_0+hs,&j>c.
 \end{cases}
\]
For $s>cg$, put $e=s-cg\in(0,w_0)$. The two selected masses are
$(M-b-1)w_0$ and $(H-j-b)w_0-e$; the corresponding lower bounds are
\[
 \begin{cases}
 [h(c-M)+(b+1)(h-c)]w_0,&j\le c,\\
 [h(j-M)+b(H-2c)+h]w_0-ce,&j>c.
 \end{cases}
\]
All four bounds are nonnegative. This proves $S_t\ge0$ and
settles $s>cg$; also $L_b=2S_b$ when $s=0$.

Now take $0<s\le cg$ in $J_b$. For a reference word $w$, let
$r(w)$ be its local row index. For $\eta\in\{0,1,2\}$, define
$\chi_\eta=\chi+\eta\mathbf1_{[cg,Q)}$ and let
$\langle\chi_\eta\rangle_F$ denote its average on a nonempty set $F$.
The reference mean gives the complementary identity
\[
 L_b+\eta S_b=(2+\eta)\rho(n_C+n_D)
       -\sum_{w\in F_C\cup F_D}\chi_\eta(r(w)).
\]
Filling $F_C$ gives a right tail. Hence $L_b\ge0$ if
$n_C=0$ or $\langle\chi\rangle_{F_C}\le2\rho$.
Otherwise $\chi=2$ on $F_A$; $n_B=0$ is immediate. The following bounds apply respectively to
$\rho=1/2$, $\rho\le3/4$ with $\langle\chi\rangle_{F_B}\ge1$,
and $\rho\ge3/4$:
\[
 \begin{aligned}
 2n_A&\ge n_A+n_B,\\
 2n_A+n_B&\ge\tfrac32(n_A+n_B),\\
 2n_C+n_D&\le\tfrac32(n_C+n_D).
 \end{aligned}
\]
In the last case, $\chi=2$ on $F_B$ gives the result directly;
otherwise $\chi\le1$ on $F_D$, and the complementary identity applies.

Thus assume $1/2<\rho<3/4$ and
$\langle\chi\rangle_{F_B}<1$. Then $F_B$ starts before $cg$,
and $\chi_\eta\le1$ on $F_D$. Moreover, $p\ge1$: otherwise
$b=M-1$, $j=H-M$, and $H-M-1<c$; since $M$ is a proper
$q$-power divisor of $H$, this forces $q=2$, $M=c=H/2$,
contrary to $\rho>1/2$. As $\chi_1\le3$ on $F_C$,
\[
 \sum_{F_C\cup F_D}\chi_1
 \le3n_C+n_D\le2(n_C+n_D),
\]
so $L_b+S_b\ge0$ for $\rho\ge2/3$.

For $1/2<\rho<2/3$, put
$F_-=\{w\in F_B:r(w)<cg\}$, $F_+=F_B\setminus F_-$,
and $n_\pm=|F_\pm|$.
If $n_+=0$, then
$n_B+n_D\le(1-\rho)W$, $n_A+n_C\ge\rho W$, and
$n_An_D\ge n_Bn_C$ imply $(1-\rho)n_A\ge\rho n_B$,
hence $L_b\ge0$.
Otherwise
\[
 n_-+n_D=(1-\rho)W,\qquad n_A+n_C+n_+=\rho W.
\]
If $n_D\ge n_-+n_C$, then
\[
 \begin{aligned}
 (1-\rho)(n_A+n_+)&\ge2\rho n_-+(2\rho-1)n_C,\\
 (1-\rho)n_A&\ge\rho n_-,\\
 3n_A+2n_+&\ge3\rho(n_A+n_B),
 \end{aligned}
\]
using $n_A\ge n_+$ and $2-3\rho>0$. Since $\chi_1=3$
on $F_A$ and $\chi_1\ge2$ on $F_+$, this proves $L_b+S_b\ge0$.

\emph{3) Two-layer count.}
Suppose $n_D<n_-+n_C$. A complete reference source layer before
$F_B$ would give $n_D\ge|\mathcal K|+n_C+s\ge n_-+n_C$.
Thus $F_B$ is in the first of the $k=HQ/d$ source layers in $J_b$,
with row image in $[0,Q/k)$. Crossing $cg>Q/3$ requires $k=2$,
so $q=2$ and $M=H/2$. Row $Q$-blocks pair physical layers
$(2r,2r+1)$; hence adjacent boundary layers share a block only for
even $p$. Since $p\ne0$, $p\ge2$.
Their indices in $J_b$ are $2b,2b+1$, and admissibility gives
$2b\ge\ell=c-1$. Consequently,
\[
 n_A=(H/2-b)w_0,\quad n_B=n_A-s,\quad n_+=(H/2-c)w_0.
\]
Since $\chi_2=4$ on $F_A$ and $\chi_2\ge3$ on $F_+$,
\begin{align*}
 L_b+2S_b&\ge4n_A+3n_+-4\rho(n_A+n_B)\\
 &=(H-2c)(4b/H-1/2)w_0+4\rho s\ge0.
\end{align*}
Here $H/3<c<H/2$ implies $H\ge8$ at the binary scale, and
$b\ge\lfloor c/2\rfloor\ge H/8$.
Each branch uses at most $p$ complete layers, so
\eqref{eq:app-regv-compensation} gives $E\ge0$.
\end{proof}

\section{Integral Slack at Standard Endpoints}
\label{app:standard-integral-slack}
For the standard promotions, we prove that $G_p$, $1\le p\le h-1$,
with $m=m_1>0$ satisfies
\begin{equation}
 \Phi(G_p)\le\lfloor H^2d/4\rfloor-1.
 \label{eq:standard-unit-slack}
\end{equation}
Use the notation of Regime V and put
\[
 X=P+N-Hd/2,\quad
 \Gamma=(e_1+v)^{\mathsf T}u-N(N+m)/d.
\]
At these endpoints $\Xi=0$. The common endpoint proof gives
\begin{align}
 H^2d/4-\Phi(G_p)&=X^2/d+\Gamma,\nonumber\\
 N+m&=(p+2)\rho d-(2\rho-1)\tau,\label{eq:unit-count}\\
 \Gamma&\ge(2\rho-1)(N\tau/d-u([0,\tau))).\nonumber
\end{align}
If $c=H/2$, then $X=pd/2$ and the square is at least $d/4\ge1$.
Suppose $c<H/2$, and set $a=H-2c\ge1$, $h=H-c$.

\subsection{Internal cutoffs}
The physical-layer masses are nondecreasing, from zero through
$|\mathcal B|,|\mathcal L|$ to the complete-layer mass $\rho d$.
Within each layer, the $Q$-block masses are also nondecreasing.
For nondecreasing masses $z_0,\ldots,z_{H-1}$ with sum $N$, their
piecewise linear prefix mass $L(t)$ satisfies
\begin{equation}
 Nt-L(t)\ge(z_{H-1}-z_0)\min(t,1-t).
 \label{eq:quantitative-tail}
\end{equation}
Indeed, a jump of size $\Delta$ at index $k$ contributes
$\Delta\min\{(H-k)t,k(1-t)\}$, which is at least
$\Delta\min(t,1-t)$. Sum over the jumps. The actual row prefix at
$\tau$ is bounded by $L(\tau/d)$, because $H\tau$ is aligned with
the source $Q$-blocks.

For $0<\tau<d$ and $p+2<H$, \eqref{eq:quantitative-tail} gives
$\Gamma\ge a\rho g\ge1$. If $p+2=H$, then
\[
 X/d=H/2-1-(a/H)(\tau/d).
\]
For $H\ge4$, this is at least $(H-2)^2/(2H)\ge1/2$.
For $H=3,c=1$, the endpoint count and $|\mathcal B|\le|\mathcal L|$
give $\rho d-|\mathcal B|\ge m/2$. Hence
$\Gamma\ge\nu/2$ and $X^2/d\ge d/36$.
Here $M\ge3$, $d\ge27$, so their sum exceeds one.

\subsection{Periodic counts at the two ends}
Put $S=u(\mathcal K)-\rho N$. Each complete $\mathcal K$ layer
maps a source $Q$-block to a row $g$-block of mass $hg$.
The set $\mathcal K$ retains the last $h$ row blocks in each period
of $H$. Full periods contribute zero to $S$; a terminal string
of $t$ complete row blocks, with $r=t\bmod H$, contributes
\begin{equation}
 S_{\rm full}=hg\{\min(h,r)-\rho r\}.
 \label{eq:unit-periodic-excess}
\end{equation}

\emph{Right boundary: $\tau=d$.} Here $\mathcal B=\varnothing$ and
$\mathcal L=[m,d)$, of mass $L_*=HM(g-\nu)<hg$.
Writing $r=pM\bmod H$ gives
\[
 \Gamma=pS,\quad S=hg\{\min(h,r)-\rho r\}
       +(\mathbf1_{\{r<h\}}-\rho)L_*.
\]
For $r<h$, $\Gamma\ge pcM(g-\nu)\ge1$.
For $r\ge h$, $S\ge\rho(hg-L_*)$. Since $hg-L_*$ is a positive
integer, it suffices that $p\ge2$. If $p=1$ and $r\ne0$, the common
$q$-power scales give $r=M<H$, hence $M\le H/2<h$, a contradiction.

\emph{Left boundary: $\tau=0$.} Here $m\le cg$, $\mathcal L=\mathcal K$, and
$\mathcal B=\mathcal K\setminus J$, where $J=[cg,cg+m)$.
Put $n=p+2\ge3$, $t_0=(H-n)M$, $r=nM\bmod H$, and
$\varepsilon=\mathbf1_{\{t_0\bmod H\ge c\}}$.
Deleting $J$ from layer $H-n$ of the $n$ complete layers removes mass
$m$ from row block $t_0$. Since $v=n\mathbf1_{\mathcal K}-\mathbf1_J$,
\begin{align}
 S&=hg\{\min(h,r)-\rho r\}-(\varepsilon-\rho)m,\nonumber\\
 \Gamma&=nS+u([0,m))-u(J).\label{eq:unit-left-cut}
\end{align}
For $n=H$, the complete row patterns on $[0,m)$ and $J$ have equal
mass; deletion changes their difference by at least $-m$.
Thus $S=\rho m$ and $\Gamma\ge(h-1)m\ge1$.

For $n<H,t_0\ge H$, both short-interval terms in
\eqref{eq:unit-left-cut} vanish. If $\varepsilon=0$, $S\ge\rho m$;
if $\varepsilon=1$, $1\le r\le h$ and
\[
 S=\frac cH(hgr-m)\ge\frac{c(H-2c)g}{H}\ge1.
\]
For the last bound, $c(H-2c)\ge H-2$ when $H$ is even, and
$c(H-2c)\ge(H-1)/2$ when $H$ is odd; correspondingly $g\ge2$ or $3$.
If $0<t_0<c$, the prefix average of the row pattern gives
$u(J)\le hm$, and
\[
 \Gamma\ge\rho(H-n)(nhgM-m)\ge1.
\]
Finally, for $c\le t_0<H$, put $z=m/g\le c$.
Here $r=H-t_0\le h$, $\varepsilon=1$, $n\ge\max(3,r)$, and
\begin{align*}
 \Gamma/g
 &\ge\frac{nc}{H}(hr-z)-h(z-h+r)_+\\
 &\ge\frac{\max(3,r)c}{H}(hr-c)-h(r-a)_+\ge1.
\end{align*}
To check the last inequality, write $H=2c+a$, $h=c+a$.
If $r\le a$, the expression is at least $3ca/H\ge1$.
For $r=2>a$, it equals or exceeds $(c^2+3c-1)/(2c+1)\ge1$.
For $r\ge3,r>a$, multiply the expression minus one by $H$:
\begin{align*}
 &c^2\{r(r-3)+2a\}
 +c\{a[r(r-3)+3a]-2\}\\
 &\hspace{4em}-a^2(r-a)-a.
\end{align*}
The coefficients of $c^2$ and $c$ are positive, so put $c=1$.
Writing $r=a+t$ leaves
\[
 (a+1)t^2+(a^2-a-3)t+(a+1)(a^2-2).
\]
For $a=1$, this is $(2t+1)(t-2)\ge0$; for $a=2$ it is
$3t^2-t+6>0$; for $a\ge3$ all coefficients are positive.
This proves \eqref{eq:standard-unit-slack} in every case.

\section{The Binary Upper-Phase Entry}
\label{app:regv-binary-bridge}
Here $H=q=2$, $c=1$, $\rho=1/2$, and $2\mid M=d/Q$.
The scale $Q=d$ is covered by the $M<H$ argument in
Proposition~\ref{prop:regv-upper-bridge}.
Write $m=m_1=a+s$, $a=bQ$, $0\le s<Q$.
Usually we interpolate through $\mathcal L_1$ or $\mathcal B_0$ to
$G_0=\mathcal L_1\mathbin{\dot\cup}\mathcal B_0$.
The boundary case
\begin{equation}
 a=d/2-Q,\qquad Q/3<s<Q
 \label{eq:binary-direct-entry}
\end{equation}
instead enters the terminal residual through $\mathcal L_0$.
All three single-layer choices are admissible and phase-aligned.
For such a choice $F$, put $N=|F|$, $I=(e_1+v_F)^{\mathsf T}u_F$,
and $\mathscr S(F)=d-\Phi(F)$. Then
\begin{equation}
 \mathscr S(F)\ge\frac{(m+N-d)^2-N(N+m)}d+I.
 \label{eq:binary-bridge-slack}
\end{equation}

\emph{The usual bridge.} If $m\ge d/2$, choose $F=\mathcal B_0$.
Its row indices lie below $d/2$, so $I=N$; since $N+m\le d$,
\eqref{eq:binary-bridge-slack} is nonnegative.
For $m<d/2$ outside \eqref{eq:binary-direct-entry}, choose
$F=\mathcal L_1$. Each source $Q$-block maps to a row $g$-block.
The successive row-block masses form a nondecreasing terminal string,
and $\mathcal K$ retains the last block of each pair. All row indices
exceed $m$, so at least half the mass lies in $\mathcal L$.
Thus $I\ge N/2$, with
\[
 N=\begin{cases}(d-a)/2,&s\le Q/2,\\
 (d-a+Q-2s)/2,&s>Q/2.
 \end{cases}
\]
Equation~\eqref{eq:binary-bridge-slack} gives the respective lower bounds
\begin{align*}
 \Psi_1(a,s)&=\frac{2a^2+(6s-3d)a+d^2-6ds+4s^2}{4d},\\
 \Psi_2(a,s)&=\frac{2a^2+(2Q+2s-3d)a+d^2-3dQ+2Qs}{4d}.
\end{align*}
Both decrease with $a$ on $0\le a\le d/2-Q$.
For $a\le d/2-2Q$, their values at $d/2-2Q$ are
\[
 \begin{gathered}
 \frac{d(2Q-3s)+4(2Q-s)(Q-s)}{4d}>0,\\
 \frac{ds+2Q(2Q-s)}{4d}>0,
 \end{gathered}
\]
respectively. At $a=d/2-Q$ and $s\le Q/3$,
\[
 \Psi_1(a,s)=\frac{d(Q-3s)+2(Q-s)(Q-2s)}{4d}\ge0.
\]
Lemma~\ref{lem:column-layer-interpolation} fills both steps to $G_0$.

\emph{Direct entry into the terminal residual.}
Under \eqref{eq:binary-direct-entry}, select $F=\mathcal L_0$.
Its row indices lie in $[0,d/2)$.
For $Q/3<s\le Q/2$, the weight $e_1+v_F$ equals one there
except on $[m,a+Q/2)$, where it vanishes.
The row multiplicity on this interval is two. Hence
\[
 N=d/4+Q/2,\qquad I=N-(Q-2s),
\]
which gives
\[
 I-\frac{N(N+m)}d
 =\frac{3s-Q}{2}
  +\frac{(d-2Q)(d-2Q+4s)}{16d}\ge0.
\]
For $Q/2<s<Q$, the weight is one throughout $[0,d/2)$, so
\[
 I=N=d/4+Q-s,\qquad N+m=3d/4.
\]
Thus $F$ is feasible in both cases. Interpolate within its layer, then
add the other layer's complete $\mathcal K$ and the current layer's
remaining incomplete groups to reach $F^{(1)}$ in
\eqref{eq:terminal-first-set}. Their columns lie inside and outside
$\mathcal K$, respectively, so the difference is column-unique.
Lemma~\ref{lem:first-packet-endpoint} gives $\Phi(F^{(1)})<d$;
integrality and Theorem~\ref{thm:matrix-interpolation} fill this step.
Continue with the terminal residual.

\section{Exceptional Terminal-Residual Scales}
\label{app:regv-qh-gt-d}

The enlarged-profile estimate needs a separate check only for
\((H,x)=(2,\tfrac12),(3,\tfrac23),(4,\tfrac34)\) with \(QH>d\).
Here \(Q=d\) for \(H=2,3\), while \(Q\in\{d,d/q\}\) for \(H=4\).
Retain the complete-group endpoint $F_{\rm V}^\star$ and add the
first residual packet. Below, $F$ denotes the cumulative set
$F^{(1)}$ in \eqref{eq:terminal-first-set}; its profiles include
all preceding selections.

\emph{The binary scale.}
For \(H=2\), \(Q=d\), and \(0<K_1\leq1/4\), the cumulative prefix
profile is \(\e+e_0\). The residual parts have masses $\mu_1$ and
$d/2-\mu_1$, with weights 1 and 2, respectively. Thus
\begin{equation}
 \Phi(F)=d-dK_1=d-\mu_1<d=\frac{q^{\lambda_3}}4.
 \label{eq:app-regv-exceptional-H2}
\end{equation}

\emph{The ternary scale.}
For \(H=3\), \(Q=d\), and \(1/3<K_1\leq1/2\), direct substitution gives
\begin{align}
 \frac{\Phi(F)}{q^{\lambda_3}}
 &\leq\frac19
 \min\{4-5K_1,-9K_1^2+8K_1\}\nonumber\\
 &\leq\frac{16}{81}<\frac14.
 \label{eq:app-regv-exceptional-H3}
\end{align}

\emph{The two \(H=4\) scales.}
Here \(x=3/4\).  If \(Q=d\) and \(1/2<K_1\leq5/8\), the same substitution
yields
\begin{equation}
 \frac{\Phi(F)}{q^{\lambda_3}}
 \leq\frac{-16K_1^2+16K_1-1}{16}<\frac{3}{16}.
 \label{eq:app-regv-exceptional-H4a}
\end{equation}
If \(Q=d/q\) and \(1/2<K_1\leq9/16\), then
\begin{equation}
 \frac{\Phi(F)}{q^{\lambda_3}}
 \leq\frac{-16K_1^2+15K_1+1/8}{16}<\frac{29}{128}.
 \label{eq:app-regv-exceptional-H4b}
\end{equation}
All four displayed bounds are strictly below \(1/4\), completing the
exceptional cases.

\end{document}